\documentclass[a4paper, 12pt, abstract=true, numbers=noenddot]{scrartcl}

\usepackage{enumitem}		
\usepackage{amssymb}		
\usepackage{amsmath}		
\usepackage{amsthm}			

\usepackage[utf8]{inputenc}
\usepackage{lmodern}
\usepackage[T1]{fontenc}
\usepackage[english]{babel}

\usepackage{dsfont} 

\usepackage{graphicx}
\usepackage{caption}
\usepackage{booktabs}

\usepackage{hyperref}

\usepackage{authblk}

\usepackage{comment}
\usepackage{makecell}

\usepackage[title]{appendix}

\title{Inference in Tightly Identified and Large-Scale Sign-Restricted SVARs}
\author{Markku Lanne}
\author{Jani Luoto}
\author{Adam Rybarczyk\thanks{Corresponding author: adam.rybarczyk@helsinki.fi \\
\noindent We thank Jonas Arias, Joshua Chan, and Juan F. Rubio-Ramírez for helpful comments. Financial support from the Research Council of Finland (grant 347986) is gratefully acknowledged}}

\affil{University of Helsinki}

\usepackage[backend=biber, style=apa, citestyle=apa, sortcites=false, maxcitenames=2]
{biblatex}
\let\temp\epsilon
\let\epsilon\varepsilon
\let\varepsilon\temp
\renewcommand{\vec}{\operatorname{vec}}

\newcommand{\tr}{\operatorname{tr}}
\newcommand{\R}{\mathbb{R}}
\newcommand{\logit}{\operatorname{logit}}
\newcommand{\Rhat}{\hat{R}}

\newcommand{\tA}{\tilde{A}}

\newtheorem{proposition}{Proposition}

\newcommand{\dif}[1]{\operatorname{d}\!#1}

\begin{document}
\maketitle




\begin{abstract}
    
    \noindent We propose a new approach to inference in tightly identified and large-scale structural vector autoregressions based on a reparameterization that enables imposing identifying inequality restrictions through continuously differentiable mappings. Permitted inequality restrictions include shape and ranking restrictions as well as bounds on economically relevant elasticities, and the approach is also able to accommodate zero restrictions in a straightforward manner. We implement a Hamiltonian Monte Carlo algorithm and show how the posterior density can be rapidly evaluated under the reparameterization, thus facilitating inference in high-dimensional settings. Two empirical applications demonstrate that our approach tends to result in lower serial dependence in Markov chains, larger effective sample sizes and reduced computation time relative to existing methods.
\end{abstract}

\newpage
    
	\section{Introduction}
	Imposing sign restrictions is a central tool for identifying structural vector autoregressions (SVARs). These restrictions are typically implemented by first estimating the posterior distribution of the reduced-form parameters and then mapping posterior draws into the structural parameter space. The most commonly used algorithm \parencite[see, for example][, and the references therein]{faust1998, canova2002,uhlig2005, rubioramirez2010} is an accept-reject sampling procedure that draws independent posterior samples of the reduced-form parameters from a conjugate normal-inverse-Wishart distribution, combines them with independent draws from the uniform distribution over orthogonal matrices, and retains only those that satisfy the imposed sign restrictions. 
    As shown by \textcite{arias2018}, this algorithm generates independent draws of the structural parameters from a well-defined posterior distribution induced by the reduced-form prior under the imposed sign restrictions. It has two appealing properties: in addition to producing independent and identically distributed posterior draws, it can readily accommodate various additional identifying restrictions, such as ranking and elasticity bounds restrictions.

    However, the performance of the accept-reject algorithm crucially depends on the size of the admissible set of orthogonal matrices satisfying the restrictions. As the restrictions become tighter, the admissible set shrinks and the probability of accepting a draw declines sharply, rendering the algorithm increasingly inefficient and, in extreme cases, infeasible. This concern is particularly relevant because  strong identifying restrictions that substantially reduce the size of the admissible set are increasingly being imposed in empirical applications \parencite[see, e.g.,][, and the references therein]{inoue2026}. In the same vein, accept–reject algorithms face challenges in large-scale SVAR models, where the parameter space is high-dimensional and identification often requires imposing a large number of sign restrictions, frequently combined with additional identifying constraints. For example, \textcite{chan2025} impose more than one hundred sign restrictions to identify eight shocks in a 35-variable SVAR. Such designs naturally entail a substantial computational burden and highlight the need for inference methods that extend beyond standard accept-reject procedures.

    Instead of relying on accept-reject algorithms, we propose imposing inequality restrictions through a reparameterization of the structural model by continuously differentiable mappings. The restrictions are thus directly incorporated into the model, which eliminates the need to discard posterior draws and substantially improves computational efficiency, especially when the admissible set is small. Our approach accommodates a broad range of inequality restrictions commonly used in the SVAR literature. These restrictions, applicable in both small and large models, include shape and ranking restrictions on impulse responses at multiple horizons, as well as bounds on the magnitude of economically relevant elasticities. The framework also allows imposing zero restrictions without requiring importance weighting or other corrections typically needed under accept-reject sampling.
    
    The flexibility of the proposed reparameterization comes at a cost, as the induced posterior distribution of the parameters does not belong to a standard family, which precludes direct sampling. Therefore, posterior inference must rely on Markov chain Monte Carlo (MCMC) methods, and we recommend the No-U-Turn Sampler (NUTS) of \textcite{hoffman2014}, a modern variant of Hamiltonian Monte Carlo (HMC) that has only rarely been employed in the SVAR literature. It exploits gradients of the log-posterior and the smooth posterior geometry induced by the reparameterization to explore the posterior effectively. However, gradient-based sampling is computationally demanding, as each iteration requires repeated evaluations of the log-posterior and its gradient. Ensuring fast computation of the posterior density is therefore crucial, and  we show how it can be evaluated efficiently under the proposed reparameterization to facilitate inference even in high-dimensional parameter spaces. For numerical work, in this paper, we use the software package Stan \parencite{stanreference}, containing an implementation of the NUTS with several desirable features, including adaptive tuning, to facilitate estimation.

    Recently, \textcite{kitagawa2025} have argued that standard MCMC algorithms do not perform well in Bayesian inference on set-identified models because set identification induces flat posterior regions along observationally equivalent parameter manifolds, which can severely impair mixing in these algorithms. However, while they show that random-walk Metropolis performs poorly in this setting, with mixing deteriorating sharply as the size of the identified set increases, those results do not necessarily carry over to modern, gradient-based samplers with adaptive tuning. Indeed, using the simulation design considered in \textcite{kitagawa2025}, we show that the No-U-Turn Sampler as implemented in Stan is able to effectively explore high-dimensional, set-identified posteriors with flat directions, producing well-mixed chains even when the identified set is large.\footnote{The simulation exercise does not actually involve an SVAR model, but in tightly identified SVARs the identification problem is less severe, which suggests that Hamiltonian Monte Carlo is particularly well suited for posterior exploration in set-identified SVAR models, and this conclusion is also supported by our empirical results.} 
    
    We are not the first to depart from standard accept–reject algorithms in sign-identified SVAR inference, as already \textcite{baumeister2015, baumeister2019} propose to use an MCMC sampler to draw the structural parameters directly. Their approach relies on conjugate prior structures, which produce closed-form conditional posteriors for all parameters except those of the contemporaneous structural relations. The latter are first drawn from their marginal posterior, after which inference for the remaining parameters proceeds using their conditional posteriors. This simplifies computation, but comes at the cost of reduced prior flexibility. Moreover, without substantial computational refinements, their algorithm can become computationally demanding in large VAR systems. More recently, \textcite{hou2024} derives the conditional posteriors for the structural parameters under linear equality and inequality restrictions that can be used to build a Gibbs sampler for an SVAR model. His approach has the limitation that for longer horizons of the IRFs, an accept-reject step is needed, which is not the case with our approach.

    \textcite{bruns2023} develop an importance sampler for sign-restricted Bayesian SVARs that uses proposal draws based on a conjugate reduced-form posterior. Their approach improves sampling efficiency and enables estimation beyond the standard conjugate prior setup. However, importance sampling methods can perform poorly in high-dimensional settings, and, as the authors note, the approach is not intended for large VAR systems. \textcite{read2025} employ slice sampling for posterior simulation and replace binary sign restrictions with a continuous relaxation. Also their algorithm involves an importance sampling step, which is used to obtain draws from the correct target distribution. Furthermore, although their method yields independent posterior draws, it is formulated under a conditionally uniform prior in the way described in \textcite{uhlig2017}, which is not invariant to the set of imposed restrictions, as shown by \textcite{arias2025elliptical}. Instead, our approach uses a uniform prior over orthogonal matrices and therefore guarantees invariance to the restrictions (see \textcite{arias2025uniform}).

    The inference methods recently put forward by \textcite{chan2025} and \textcite{arias2025elliptical} are particularly useful in tightly identified and large-scale SVARs. \textcite{chan2025} exploits the fact that a uniform distribution over the orthogonal matrices is preserved under reordering the columns or flipping their signs and thus yields multiple parameter draws at a lower computational cost in an accept-reject algorithm. \textcite{arias2025elliptical}, in turn, propose embedding elliptical slice sampling within a Gibbs sampler, which allows draws of the orthogonal matrix to remain within the identified set by construction. Both approaches deliver substantial computational gains in tightly identified and large-scale SVARs, while maintaining the conventional reduced-form prior structure. However, because of reliance on an accept–reject framework, the approach of \textcite{chan2025} is not able to achieve the same level of performance as that of \textcite{arias2025elliptical}.

    As the method of \textcite{arias2025elliptical} represents  the current state of the art in inference in tightly identified and large-scale SVAR models, we focus on benchmarking our approach against theirs in Section \ref{sec:empirical}. In particular, we revisit their empirical applications to an oil market SVAR identified by sign and elasticity restrictions, as well as a 35-variable SVAR model in which ten shocks are identified by sign and ranking restrictions. In both settings, our HMC approach exhibits faster convergence, as measured by the $\Rhat$ statistic of \textcite{vehtari2021}. While differences are small in the oil market SVAR, with our method only doubling the speed of convergence, they are substantially greater in the 35-variable model. In that setting, our approach  converges to the stationary distribution approximately five times faster than the method of \textcite{arias2025elliptical}. Measured by the effective sample size (ESS), the number of independent draws equivalent to the correlated MCMC samples, it also attains substantially higher sampling efficiency, with ESS up to an order of magnitude greater than those of the elliptical slice sampler. These results are likely to follow from the fact that the Gibbs sampler of \textcite{arias2025elliptical} generates Markov chains with extremely persistent serial correlation, which remains sizeable even at lags in the hundreds, whereas autocorrelations in our sampler decay comparatively quickly.
    
    While our sampler performs better in the two applications, the approach of \textcite{arias2025elliptical} offers complementary advantages. In particular, in contrast to their sampler, our approach does not allow for arbitrary combinations of inequality restrictions; for example, in general, long-run and dynamic sign restrictions cannot both be imposed simultaneously. On the other hand, as shown in Section \ref{sec:35variable}, we are able to impose exclusion restrictions (even at no additional computational cost) which their approach does not afford. Our method also supports highly flexible prior specifications, which need not yield closed form conditional posteriors, although this may come at a computational cost.    

    Finally, it is worth noting that our paper is also related to \textcite{plagborgmoeller2019}, who estimates vector moving average (VMA) models using the NUTS to sample directly from a moving-average representation. In our SVAR setting, we adopt a mixed parameterization: the moving-average coefficient matrices subject to the identifying restrictions are sampled directly, whereas the remaining dynamics are parameterized in terms of reduced-form VAR coefficients. In our experience, this parametrization is well suited to the NUTS, with chains mixing well and converging rapidly even under diffuse priors. Furthermore, this parameterization allows for the usage of a variety of priors designed for VAR models. By contrast, posterior simulation in unrestricted VMA models is computationally challenging, and satisfactory mixing and convergence may be difficult to achieve without strongly informative priors, which can be restrictive in empirical applications.

    The remainder of the paper is organized as follows. In Section \ref{sec:model}, we present the SVAR model. Section \ref{sec:id_restrictions} introduces the identification problem and the restrictions that can be implemented in our framework. In Section \ref{sec:bayesian_inference}, we describe the prior and posterior, our method of imposing identifying restrictions, and an efficient method of computing the posterior density, which is crucial for the implementation of our approach in tightly identified and large-scale SVAR models. Section \ref{sec:empirical} illustrates the performance of our approach in empirical applications. Finally, Section \ref{sec:conclusion} concludes.

	\section{Model}\label{sec:model}
	We consider the SVAR model
	\begin{align}
		A_0 y_t &= c + A_{1} y_{t-1} + \dots + A_{p} y_{t-p} + \epsilon_{t}, \label{eq:svar_a}
	\end{align}
	where $y_t$ is an $N$-vector of observables, $c$ is an $N$-vector of constants, $A_i$, $i=1,\dots p$, are $N\times N$ autoregressive coefficient matrices, and the $N\times N$ matrix $A_0$ summarizes the contemporaneous structural relations among the elements of $y_t$. The $N$-vector of structural shocks $\epsilon_{t}$ has  mean zero and identity covariance matrix $I_N$, and its elements are independent across time. We assume that $\epsilon_{t}$ is jointly Gaussian, but this assumption can be relaxed, for example, if we wish to make use of non-Gaussianity in identification or to place restrictions on higher moments.

	The SVAR model \eqref{eq:svar_a} is in the so-called A-form (see \cite{Lutkepohl2005}, Chapter 9), but it can equivalently be written in the B-form
	\begin{align}
		y_t &= \tilde{c} + \tilde{A}_{1} y_{t-1} + \dots + \tilde{A}_{p} y_{t-p} + B \epsilon_{t}, \label{eq:svar_b}
	\end{align}
	where $\tilde{c} = A_0^{-1}c$, $\tilde{A}_i = A_0^{-1}A_i$, and $B = A_0^{-1}$. This representation is useful for estimation because the matrix $B$ gives us the impact responses directly, whereas $A_0$ needs to be inverted to obtain them.
	
	The usual objects of interest in SVAR analysis are the structural impulse response functions (IRF), which trace out the responses of observables at different time horizons $t+i$, $i=0,1,2,\dots$, for a shock occurring at time $t$. The structural IRFs are obtained recursively as $\Psi_0 = A_0^{-1} = B$, and
	\begin{align}
		\Psi_i &= \sum_{j=1}^{\min\{i,p\}} \tilde{A}_j\Psi_{i-j}. \label{eq:irf_recursion}
	\end{align}
	The $(l,k)$ element of $\Psi_i$, $\Psi_{i,lk}$, is the response of variable $l$ to shock $k$ at horizon $i$.

	\section{Identifying Restrictions} \label{sec:id_restrictions}

    This section describes the types of identifying restrictions that can be accommodated within our framework. In Section \ref{sec:enforcing_restrictions}, we illustrate how these restrictions can be implemented in practice.
    
    In the reduced-form vector autoregression (VAR) corresponding to model \eqref{eq:svar_a},
	\begin{align}
		y_t &= \tilde{c} + \tilde{A}_{1} y_{t-1} + \dots + \tilde{A}_{p} y_{t-p} + u_{t},
	\end{align}
    $u_t =  A_0^{-1} \epsilon_{t}$ are the reduced-form errors with mean zero and covariance matrix $\Sigma = A_0^{-1}(A_0^{-1})'$. The structural model suffers from the well-known identification problem that the covariance matrix of the reduced-form residuals is identified, but any $\bar{A}_0$ satisfying $\bar{A}_0^{-1} =  A_0^{-1}Q$, for an orthogonal matrix $Q$, yields the same $\Sigma$, that is, the same covariance structure of $u_t$. Therefore, the corresponding models are observationally equivalent.
    
	To identify the structural parameters of an SVAR model, the researcher must use prior information, typically in the form of identifying restrictions. Currently, sign restrictions are probably the most widely used type of such restrictions. In the implementation of \textcite{uhlig2005, rubioramirez2010}, candidate draws of the impact matrix are generated from the product $PQ$, where $P$ is the Cholesky factor of the posterior covariance matrix $\Sigma$ of $u_t$ and $Q$ is an orthogonal rotation matrix drawn from its prior distribution, and only draws satisfying the sign restrictions are retained. However, it is well known that sign restrictions do not point-identify the model, and hence, the prior is not updated within the identified set in Bayesian estimation  \parencite[see, e.g.,][]{baumeister2015}. This problem can be avoided by sufficiently tight identification of the orthogonal rotation matrix $Q$, because in that case posterior uncertainty is mostly attributed to uncertainty about the reduced-form parameters \parencite{inoue2020}.
	
	However, tight identification of $Q$ poses computational challenges, as the probability of randomly drawing a candidate $Q$ satisfying all restrictions imposed on the model is inversely proportional to the size of the identified set. Hence, estimation time can be prohibitively long when the restrictions are tight and/or there are a lot of them. To that end,  several solutions have recently been proposed to increase the efficiency of sampling the orthogonal matrix $Q$. First, \textcite{chan2025} exploit the fact that permutations of columns and flips of their signs preserve a uniform distribution over $Q$, which facilitates obtaining multiple candidate draws from a single $Q$. Second, \textcite{arias2025elliptical} use elliptical slice sampling, which is a particular Markov Chain Monte Carlo (MCMC) algorithm, to sample $Q$ such that the draws remain within the identified set. Finally, \textcite{read2025} approximate sign restrictions by a continuous function, which allows more efficient sampling using the slice sampler. However, their procedure requires an additional importance sampling step to ensure that the draws follow the correct distribution, and it is formulated under a conditionally uniform prior, which is not invariant to the set of imposed restrictions, as discussed in the Introduction.
    
	Instead of sampling $Q$, the method proposed in this paper works with the impact matrix $B$ directly, which facilitates imposing  sign restrictions on the parameters of interest using continuously differentiable mappings to transform them into an unconstrained parameter space. Thus, there is no need to perform any rejection sampling. This approach works well even in very tightly identified models, as regardless of the size of the identified set, the sampler always runs on an unrestricted parameter space. Our method allows for sign restrictions on the impulse responses also beyond the impact effect, as long as the horizon to be restricted  does not exceed the VAR lag length $p$. We believe that this limitation on the horizon is unlikely to be important in most empirical applications, especially in light of the recent work by \textcite{olea2025} recommending relatively long lag lengths.

    Besides sign restrictions, our approach facilitates imposing other kinds of identifying restrictions. In contrast to \textcite{arias2018}, exclusion restrictions can be imposed without an importance sampling step, thus avoiding computation of volume elements. Also, the effect of a given shock can be constrained to increase over a specified horizon, or a shock can be required to have the largest impact effect on a given variable. The exact method of placing the latter kinds of restrictions will be presented in Section \ref{sec:enforcing_restrictions}. Finally, sign restrictions can be combined with other kinds of restrictions, including restrictions on elasticities, as in \textcite{kilian2012}, who impose them to avoid implausibly large estimates of the short-run elasticity of oil supply under sign identification alone. More generally, we can impose ranking restrictions requiring a shock to have a greater or smaller effect on a given variable than on another variable (see \cite{amirahmadi2021}). The use of these kinds of restrictions is illustrated in Sections \ref{sec:oilmarket} and \ref{sec:35variable}, respectively.

    Our approach is flexible, but there are certain kinds of restrictions and combinations of restrictions that it does not permit, and in these cases, the samplers of \textcite{arias2025elliptical} or \textcite{chan2025} are a viable choice. In particular, because the sampler operates on IRF parameters, restrictions can only be stated in terms of the IRFs; constraints that jointly involve IRF values and VAR matrices are therefore not feasible. For example, long-run restrictions generally cannot be combined with dynamic sign restrictions on the IRFs. More broadly, only restrictions that admit a diffeomorphism\footnote{A diffeomorphism is a continuously differentiable bijection with a continuously differentiable inverse. This permits a wide range of restrictions, but, for example, restrictions on forecast error variance decompositions (FEVD) are not compatible because the squaring operations required for their calculation are not bijective.} to an unconstrained parameter vector are allowed. It is also not possible to impose restrictions on the likelihood, such as the narrative restrictions of \textcite{antolindiaz2018}.
	
	\section{Bayesian Inference} \label{sec:bayesian_inference}
	
	\subsection{Prior and Posterior}

    Following standard practice in the Bayesian VAR literature, we impose a conjugate normal–inverse Wishart (N-IW) prior on the reduced-form parameters $\tilde{A}=\left[\tilde{c}, \tilde{A}_1,\dots,\tilde{A}_p\right]'$, and the covariance matrix $\Sigma$. This choice yields a posterior distribution of the same form and improves computational efficiency. The proposed method, however, does not rely on conjugacy and remains valid under more general prior specifications. 
    
    Specifically, we assume $\Sigma \sim IW(\nu_0, S_0)$ and $\vec(\tilde{A}) \mid \Sigma \sim N(\vec(\Phi_0), \Sigma \otimes \Omega_0)$, where $IW(\cdot,\cdot)$ and $N(\cdot,\cdot)$ denote the inverse Wishart and multivariate normal distributions, respectively, and $\vec$ is matrix vectorization. The quantities $\nu_0, S_0, \Phi_0$ and $\Omega_0$ are prior parameters. Under these assumptions, the posterior distribution is again normal–inverse Wishart, with updated parameters given by
	
    \begin{equation}
        \begin{gathered}
            \tilde{\nu} = \nu_0 + T \\
            \tilde{\Omega} = (X'X + \Omega_0^{-1})^{-1} \\
            \tilde{\Phi} = \tilde{\Omega}(X'Y + \Omega_0^{-1}\Phi_0) \\
            \tilde{S} = Y'Y + S_0 + \Phi_0'\Omega_0^{-1}\Phi_0 - \tilde{\Phi}' \tilde{\Omega}^{-1} \tilde{\Phi}
        \end{gathered}
            \label{eq:posterior_parameters}
    \end{equation}
	
	The rotation matrix $Q$ is assigned a uniform prior over the space of orthogonal matrices, independent of the reduced-form parameters. Because the likelihood is invariant to orthogonal rotations, the posterior distribution of $Q$ coincides with the prior.
	
    Our aim is to sample from the structural parameters of the model, which requires a transformation of the posterior density of the orthogonal reduced-form parameters into the implied density over the structural parameters of interest. Related transformations from the orthogonal reduced-form parameters to the structural parameters have been derived in previous work. For instance, \textcite{arias2018} characterize the induced distribution of the structural parameters in the A-model, referring to it as the normal–generalized normal distribution. \textcite{inoue2013} derive the distribution of the first $p+1$ structural impulse response matrices, and \textcite{inoue2022} provide a further generalization by allowing the number of IRF matrices to differ from $p+1$.
    
    While we build on \textcite{inoue2013, inoue2022}, we adapt their results to our setting in which the posterior is parameterized in terms of the first $k$ impulse response matrices $\Psi_0,\Psi_1,\dots,\Psi_k$ (on which the identifying restrictions are imposed) and the remaining VAR coefficient matrices $(\tilde{A}_{k+1},\dots,\tilde{A}_p)$. As our interest is in estimating a full SVAR model, we keep the VAR matrices in the density, as opposed to \textcite{inoue2022} who only consider the IRF parameters. This mixed parameterization also improves the efficiency of posterior simulation because the VAR coefficients tend to exhibit weaker posterior dependence than impulse response parameters. In addition, by exploiting the properties of the $LQ$ decomposition, we obtain a simpler expression for the associated Jacobian determinant than \textcite{inoue2013}. It does not include the orthogonal matrix $Q$, which also saves on calculations as there is no need to evaluate a determinant of another large matrix.

    The following proposition characterizes the posterior of the structural parameters collected in $\Pi_k=(\tilde{c},B,\Psi_1,\dots,\Psi_k,\tA_{k+1},\dots,\tA_p)$ in terms of the N-IW posterior of the orthogonal reduced-form parameters.
    
    \begin{proposition}
        \label{prop:structural_density}
        Let $NIW_{(\nu, \Omega, \Phi, S)}(\Sigma, \tilde{A})$ denote the normal-inverse Wishart density with parameters $\nu, \Omega, \Phi, S$, s.t.
        \begin{align}
            NIW_{(\nu, \Omega, \Phi, S)}(\Sigma, \tilde{A}) \propto |\Sigma|^{-\frac{\nu+N+1}{2}} e^{-\frac{1}{2}\tr(S\Sigma^{-1})} |\Sigma|^{-\frac{Np+1}{2}} e^{ -\frac{1}{2} \vec(\tilde{A}-\Phi)' (\Sigma\otimes\Omega)^{-1} \vec(\tilde{A}-\Phi) }.
        \end{align}
        Then the posterior of structural parameters $\Pi_k$ satisfies
        
        \begin{align}
        	p(\Pi_k\mid Y) \propto NIW_{(\tilde{\nu}, \tilde{\Omega}, \tilde{\Phi}, \tilde{S})}(\Sigma(B), \tilde{A}(\Pi_k)) \cdot |B|^{-kN+1}, \label{eq:unrestricted_posterior}
        \end{align}
        where $\tilde{\nu}, \tilde{\Omega}, \tilde{\Phi}, \tilde{S}$ are given by equation \eqref{eq:posterior_parameters}.
    \end{proposition}
    
    \begin{proof}
    See Appendix \ref{app:structural_density}.
    \end{proof}
    The proof of Proposition \ref{prop:structural_density} does not rely on the distribution over the reduced-form parameters being N-IW, and thus any other prior or error distribution can be specified as well. For example, the asymmetric conjugate prior of \textcite{chan2022} could be used with an appropriate Jacobian adjustment. However, we state this proposition using the N-IW density for compatibility with much of the existing literature, and because it allows fast estimation of large SVARs, as will be demonstrated in Section \ref{sec:35variable}
    
    Our goal is to incorporate both inequality and zero restrictions, and thus we also need a result that gives us the density function under these restrictions. The following proposition is a direct result of standard facts, but we state it here for clarity.

    \begin{proposition}
    \label{prop:conditional_density}
        Suppose the random vector $\theta$ has a continuous density over $\R^{n_1 + n_2 + n_3}$, and partition it as $\theta = (\theta_1, \theta_2, \theta_3)$, with the dimensions of $\theta_i$ being $n_i$. Let $R\subset \R^{n_1}$ be a set of positive probability. Suppose that $c\in\R^{n_2}$ is a vector of constants, and $\int_{R} \int_{\R^{n_3}} f(\theta_1, c, \theta_3)\dif{\theta_3}\dif{\theta_1} > 0$.
        Then the conditional density of $\theta$ satisfies
        \begin{align}
            f(\theta_1, \theta_2, \theta_3 \mid \theta_1 \in R, \theta_2 = c) &\propto \mathds{1}_{\{\theta_1 \in R\}} f(\theta_1, c, \theta_3)
        \end{align}
    \end{proposition}
    \begin{proof}
        See Appendix \ref{app:conditional_density}.
    \end{proof}
    In the context of our SVAR model, the density $f$ is the density of the structural parameters $\theta$ in \eqref{eq:unrestricted_posterior}. The set-identifying restrictions (like sign and ranking restrictions) and exclusion restrictions are imposed on $\theta_1$ and $\theta_2$, respectively, while $\theta_3$ collects all free parameters of the model. The key point of Proposition \ref{prop:conditional_density} is that evaluation of the conditional density under the exclusion restrictions requires no normalizing constants. Therefore, in the case of a Metropolis-style algorithm like HMC, it is sufficient to evaluate the unconstrained density on the subspace satisfying the restrictions. It is important to point out that Proposition \ref{prop:conditional_density} is more general than our method affords. In particular, we cannot restrict $\theta_1$ in a totally arbitrary way, as discussed at the end of Section \ref{sec:id_restrictions}. We also confine equality restrictions to zero restrictions to preculde the possibility of disjoint posteriors \parencite[see][]{kitagawa2025}.

	\subsection{Imposing Identifying Restrictions} \label{sec:enforcing_restrictions}
	
    In this section, we show how to impose identifying restrictions on the parameters using continuously differentiable transformations, which result in an unrestricted parameter vector. Our sampler runs on this vector, which means that rejection sampling is not needed to impose the identifying restrictions. Implementation of exclusion restrictions is straightforward in that no Jacobian adjustment is needed, whereas restrictions based on inequalities do require it as they are based on a change of variables. Because a different set of identifying restrictions is considered in each application, there are no general transformation and Jacobian that would cover all possible cases, but we discuss a few building blocks which can be used to construct frequently encountered restrictions.
	
	Exclusion restrictions can be imposed by setting the values of the parameters of interest to zero during the evaluation of the density function, as shown in Proposition \ref{prop:conditional_density}. For example, if prior economic knowledge dictates that the response of the first variable to the second shock on impact is zero, we set $B_{1,2} = 0$, where $B_{i,j}$ denotes the $(i,j)$ element of the matrix $B$, and drop the corresponding parameter from the parameter vector. Hence, the posterior density is evaluated only on the subspace where the restriction holds.
	
	Imposing inequality restrictions calls for suitable transformations of the parameters and the associated Jacobian determinant denoted by $|J|$. Combining the Jacobian with the density of the structural parameters in \eqref{eq:unrestricted_posterior}, yields the density over the unrestricted auxiliary vector $\theta$:
	\begin{align}
		p(\theta) \propto  NIW_{(\tilde{\nu}, \tilde{\Omega}, \tilde{\Phi}, \tilde{S})}(\Sigma(\theta),\tA(\theta))\cdot|B(\theta)|^{-kN+1} |J(\theta)|, \label{eq:auxiliary_density}
	\end{align}
	This is the density that our algorithm samples from. The exact cost of computing $|J|$  varies depending on the number and type of restrictions imposed on the model, but in applications involving sign restrictions or elasticity bounds it is negligible in comparison to the cost of evaluating the normal-inverse Wishart density. For example, if we restrict the sign of every parameter in $B$, evaluation of $\log|J|$ collapses to computing the sum over the $N^2$ elements of $\theta$ corresponding to $B$.
    
    Let us consider the implementation of typical inequality restrictions. Suppose first that we have a density $p_\alpha(\alpha)$ over a parameter $\alpha\in\R_+$, which could, for example, be an impact response parameter from the matrix $B$. A Metropolis-type sampler running on $\alpha$ produces proposals $\alpha < 0$, which are always rejected, and this decreases efficiency. Instead, we can reparameterize the sampler in terms of an unrestricted parameter, so that no a priori known regions of zero probability are encountered. To that end, we introduce a parameter $\tilde{\alpha}\in\R$, and set $\alpha = e^{\tilde{\alpha}}$. The Jacobian of this transformation is $e^{\tilde{\alpha}}$, and the corresponding density function is then $p(\tilde{\alpha}) = p_\alpha(e^{\tilde{\alpha}})e^{\tilde{\alpha}}$ by a change of variables. If we want to restrict $\alpha < 0$, we set $\alpha = -e^{\tilde{\alpha}}$, and the corresponding density is $p(\tilde{\alpha}) = p_\alpha(-e^{\tilde{\alpha}})e^{\tilde{\alpha}}$.
	
	Another useful transformation is the logit transformation for scalars with lower and upper bounds. Suppose we want to sample a parameter $\alpha\in(a,b)$, say, an elasticity parameter which needs to be positive ($a=0$) and also constrained from above.  To that end, we define an $\tilde{\alpha}\in\R$, and obtain $\alpha$ as
	\begin{align*}
		\alpha = a + (b-a)\logit^{-1}(\tilde{\alpha}),
	\end{align*}
	where
	\begin{align*}
		\logit^{-1}(\tilde{\alpha}) = \frac{1}{1+e^{-\tilde{\alpha}}}.
	\end{align*}
	In this case, the absolute value of the Jacobian is
	\begin{align*}
		\left| \frac{d}{d\tilde{\alpha}} \alpha(\tilde{\alpha}) \right| &= (b-a)\logit^{-1}(\tilde{\alpha})(1-\logit^{-1}(\tilde{\alpha})),
	\end{align*}
	and the corresponding density function is obtained by a change of variables, as in the sign-restricted case above. Notice that the upper and lower bounds can be parameters themselves, in which case the dependencies need to be accounted for when calculating the Jacobian.
	
	Besides sign and bound restrictions, ordering restrictions related to the relative sizes of the effects of the shocks are often imposed. For instance, we might want to declare a priori that a certain shock has a greater impact effect on a given variable than any other shock. As an example in the context of monetary policy, we could require that the impact effect of a contractionary monetary policy shock on the interest rate be positive and greater in absolute value than that of any other shock.\footnote{The impact effect must be greater in absolute value, not just greater than the impact effect of any other shock because otherwise the estimated effect of the monetary policy shock may change sign and switch columns. In such a case, the column assigned to the monetary policy shock could correspond to a shock with a small estimated effect in magnitude, yet greater than that of the estimated monetary policy shock.}
    
    These ordering restrictions can be implemented by imposing restrictions on a row of the impact matrix $B$. In the case of the monetary policy example, let $(\alpha_1,\dots, \alpha_N)$ denote the row of $B$ corresponding to the interest rate, where $\alpha_N$ denotes the impact effect of the monetary policy shock. Formally, we sample a vector $(\alpha_1,\dots, \alpha_N)\in\R^N$ subject to the restrictions
	\begin{align*}
		\alpha_N &> 0 \\
		-\alpha_N < \alpha_i &< \alpha_N \ \forall i\in\{1,\dots, N-1\}.
	\end{align*}
	By defining $\tilde{\alpha}\in\R^N$, the transformation can be written componentwise as
	\begin{align*}
		\alpha_N &= e^{\tilde{\alpha}_N} \\ \nonumber
		\alpha_i &= -\alpha_N + (\alpha_N - (-\alpha_N)) \logit^{-1}(\tilde{\alpha}_i) \\ \nonumber
		&= \alpha_N(2\logit^{-1}(\tilde{\alpha}_i)-1) \\
		&= e^{\tilde{\alpha}_N}(2\logit^{-1}(\tilde{\alpha}_i)-1)
	\end{align*}
	The partial derivatives of $\alpha_N$ needed for the Jacobian are
	\begin{align*}
		\frac{\partial \alpha_N}{\partial \tilde{\alpha}_N} &= e^{\tilde{\alpha}_N} & \frac{\partial \alpha_N}{\partial \tilde{\alpha}_i} &= 0 \:\forall i\in\{1,\dots,N-1\}.
	\end{align*}
	For the remaining elements $i\in\{1,\dots, N-1\}$, we have
	\begin{align*}
		\frac{\partial \alpha_i}{\partial \tilde{\alpha}_N} &= e^{\tilde{\alpha}_N}\left( 2\logit^{-1}(\tilde{\alpha}_i) -1 \right) = \alpha_i \\
		\frac{\partial \alpha_i}{\partial \tilde{\alpha}_i} &= 2e^{\tilde{\alpha}_N}\logit^{-1}(\tilde{\alpha}_i) \left( 1-\logit^{-1}(\tilde{\alpha}_i) \right) \\
		\frac{\partial \alpha_i}{\partial \tilde{\alpha}_j} &= 0,\ i\neq j,\ j\neq N.
	\end{align*}
	Thus, the Jacobian matrix has non-zero entries only on the diagonal and in the last column, implying that the determinant is just the product of the diagonal elements,
	\begin{align*}
		|J(\tilde{\alpha})| &= e^{N\tilde{\alpha}_N} \cdot 2^{N-1} \cdot \prod_{i=1}^{N-1} \logit^{-1}(\tilde{\alpha}_i)\left( 1-\logit^{-1}(\tilde{\alpha}_i) \right).
	\end{align*}
	If the positive element is assigned to a position other than $N$, the absolute value of the Jacobian determinant remains unchanged, as a permutation only multiplies the Jacobian by $\pm 1$.

    In an application one would decide on the set of the identifying restrictions imposed on the model, and then implement them using a combination of the aforementioned transformations or other transformations that are feasible in that particular setting. The Jacobian of these transformations is then incorporated into the density function in \eqref{eq:auxiliary_density}.
	
	To illustrate these methods in practice, let us consider a two-variable model with a single lag and no constant. Suppose that the variables are quantity and price, and the supply and demand shocks are identified using sign restrictions on the impact responses of the variables. The responses have opposite signs for the supply shock and the same sign for the demand shock. If the first shock is the supply shock, the first column of the impact matrix $B$ thus contains one positive and one negative element, whereas the second column contains two positive elements.
	
	To sample from the model, we define an unrestricted auxiliary vector $\theta\in\R^8$. The first four entries of this vector correspond to the elements of the impact matrix $B$, and the last four entries correspond to the elements of the VAR matrix $\tilde{A}_1$. As no restrictions are imposed on the elements of $\tilde{A}_1$, they enter  $\theta$ as such. Specifically,
	\begin{align*}
		\theta = (\theta_{11}^B, \theta_{21}^B, \theta_{12}^B, \theta_{22}^B, \tilde{a}_{11}, \tilde{a}_{21}, \tilde{a}_{12}, \tilde{a}_{22}).
	\end{align*}
	To impose the sign restrictions, we use the exponential transformation. Thus, the structural parameter matrices become
	\begin{align*}
		B(\theta) &= \begin{bmatrix}
			\exp({\theta_{11}^B}) &  \exp({\theta_{12}^B}) \\
			-\exp({\theta_{21}^B}) & \exp({\theta_{22}^B})
		\end{bmatrix}, & 
		\tilde{A}_1(\theta) &= \begin{bmatrix}
			\tilde{a}_{11} & \tilde{a}_{12} \\
			\tilde{a}_{21} & \tilde{a}_{22}
		\end{bmatrix},
	\end{align*}
	and the Jacobian matrix $J$ is diagonal and given by
	\begin{align*}
		J(\theta) = \begin{bmatrix}
			\exp({\theta_{11}^B}) \\
			& \exp({\theta_{21}^B}) \\
			&&-\exp({\theta_{12}^B}) \\
			&&&\exp({\theta_{22}^B}) \\
			&&&& I_4
		\end{bmatrix},
	\end{align*}
	where $I_4$ denotes the $4\times 4$ identity matrix. Therefore, the absolute Jacobian determinant is
	\begin{align*}
		|\det(J(\theta))| = \prod_{i,j} \exp({\theta_{ij}^B}).
	\end{align*}
	Given  the Jacobian and the matrices $B,\tilde{A}_1$, we can now proceed to evaluate the density using formula.\eqref{eq:auxiliary_density}.

	\subsection{Estimation Algorithm} \label{sec:algorithm}

    Because the density function \eqref{eq:auxiliary_density} does not facilitate direct sampling, we employ an MCMC method. Specifically, we use the No-U-Turn Sampler (NUTS) of \textcite{hoffman2014}, which is a variant of the Hamiltonian Monte Carlo (HMC) algorithm. An HMC algorithm explores the target distribution by simulating Hamiltonian dynamics, which allows the sampler to construct long trajectories while yielding a high acceptance probability by utilizing the gradients of the target log-density. At a high level, the HMC sampler can be thought of as a Metropolis-Hastings algorithm with a complicated proposal distribution. Its performance depends on tunable parameters $\varepsilon, L$ and $M$. 
    Starting from the current draw, the HMC sampler simulates a trajectory by integrating Hamiltonian dynamics using $L$ steps of size $\varepsilon$ with a mass matrix $M$, whose optimal choice is the inverse of the covariance matrix of the target posterior distribution \parencite[Section 4.3]{betancourt2018}. The resulting endpoint is then used as a proposal.

    The NUTS solves the issue of tuning $L$ by dynamically expanding the path taken within each iteration until it starts to make a U-turn towards the initial state. The implementation of the NUTS in Stan further provides routines for adapting the step size and metric matrix during warmup, and calculates the gradient of the target log-density by automatic differentiation. These properties reduce the scope of implementation error and ensure that sampling can be done efficiently in a wide variety of situations. Further details on the HMC sampler and NUTS can be found in Appendix \ref{app:algorithm}.

    Beyond the built-in efficiency gains provided by Stan, sampling performance can be further improved through appropriate parameterization and initialization. As discussed in detail in Appendix \ref{app:non-centered}, if only impact restrictions are considered, sampling efficiency can be enhanced by adopting a non-centered parameterization, where we sample a standard normal (instead of a general multivariate normal) vector \parencite[see][]{betancourt2013}. Moreover, finding reasonable initial values is of utmost importance, and in the applications of Section \ref{sec:empirical} we initialize the sampler by drawing from the normal-inverse Wishart posterior and constructing an orthogonal matrix compatible with the imposed restrictions one column at a time, following the method of \textcite[Appendix I]{arias2025elliptical}.
    
    It is important to note that set identification poses challenges that persist even under efficient MCMC implementations. In particular, it induces flat posterior regions along observationally equivalent parameter manifolds. Because the likelihood function provides no information in these directions, local MCMC updates receive little guidance on how to move within the identified set, which can hinder effective posterior exploration. \textcite{kitagawa2025} formally show that this issue can be severe for random-walk Metropolis algorithms, with mixing deteriorating as the size of the identified set increases. However, the mechanism underlying this result does not necessarily extend to Hamiltonian Monte Carlo. This is the case because, in contrast to random-walk Metropolis, HMC algorithms rely on Hamiltonian trajectories rather than diffusive local proposals, and therefore they explore flat posterior regions differently. Moreover, modern HMC implementations (including the one used in this paper) employ automatic step-size and mass-matrix adaptation during warmup, which further helps mitigate these exploration difficulties.
     
    To demonstrate that an HMC sampler can effectively explore the posterior under set identification, we revisit a simulation experiment that  \textcite{kitagawa2025} considered to illustrate the opposite. This example does not involve a VAR model, but we adopt it because it features a known and simple posterior geometry, which allows for a transparent diagnosis of convergence issues. Specifically, we generate data from a normal distribution with mean  $\mu=\sum_{i=1}^{k} \mu_i$ and variance unity, where the parameters $\mu_i$, $i=1,\dots,k$, are not identified for $k \geq 2$. We study the performance of Hamiltonian Monte Carlo in this setting for $k = 2$ and $k = 1{,}000$, with the latter case reflecting the high-dimensional parameter spaces encountered in VAR applications. For $k=2$ and $k=1{,}000$, we generate 1{,}000 observations from $N(2,1)$ and $N(10,1)$ distributions, respectively. The larger mean in the latter case reflects the greater number of additive components.
	
	The left panel of Figure \ref{fig:unidentified_normal1} presents a trace plot based on 12{,}000 draws for the case $k=2$, of which 2{,}000 are warmup draws. Following \textcite{kitagawa2025}, the warmup draws are included. The right panel depicts the histograms of the obtained sample (without the warmup iterations) together with the analytical marginal densities, when the support of each $\mu_i$ is bounded between $-10$ and 10.\footnote{In the case $k=2$, the marginal density is proportional to $\Phi(\sqrt{n}(\mu_i+10-\bar{x})) - \Phi(\sqrt{n}(\mu_i-10-\bar{x}))$, where $\Phi$ is the standard normal CDF, $n$ is the sample size, and $\bar{x}$ is the sample mean. For $n=1{,}000$ this density is approximately uniform.} The plots indicate that the NUTS (as implemented in Stan) is able to traverse the entire posterior distribution of $\mu$ without apparent difficulties and produces draws that closely match the analytical marginal densities. This contrasts with the behavior documented by \textcite{kitagawa2025} for a standard HMC algorithm in the same set-identified setup, where the sampler failed to adequately explore the parameter space associated with $(\mu_1,\mu_2)$.
    
	\begin{figure}[htbp]
		\centering
		\includegraphics[width=0.8\textwidth]{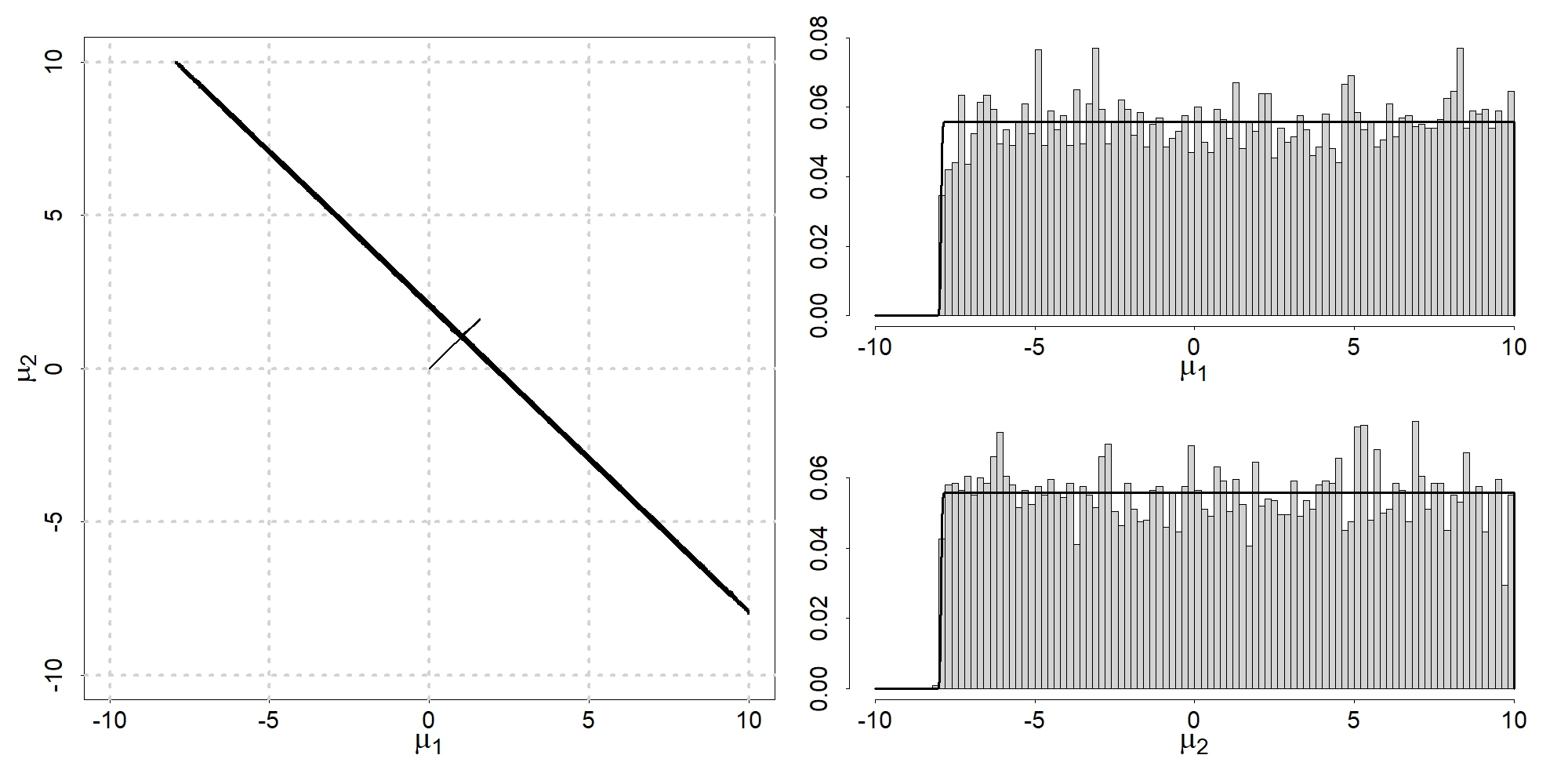}
		\caption{Simulation results for the $k=2$ case. The left and right panels depict the traceplot of $(\mu_1, \mu_2)$ and marginal distributions of $\mu_1$ and $\mu_2$, respectively. The warmup iterations are included in the traceplot but not in the histograms. The solid lines overlaid on the histograms depict the analytical marginal densities.}
		\label{fig:unidentified_normal1}
	\end{figure}
	
	The NUTS performs well even in a high-dimensional setting, as illustrated in Figure \ref{fig:unidentified_normal2} that presents a scatter plot of the first two components of $\mu$ and the estimated autocorrelation functions for these two parameters when $k = 1{,}000$. The draws are dispersed without visible clustering, and the autocorrelations decay rapidly toward zero. Corresponding plots for other pairs of parameters look essentially identical. These results indicate that, even in this high-dimensional setting, the NUTS produces high-quality samples and effectively explores the full target distribution.
	
	\begin{figure}[htbp]
		\centering
		\includegraphics[width=0.8\textwidth]{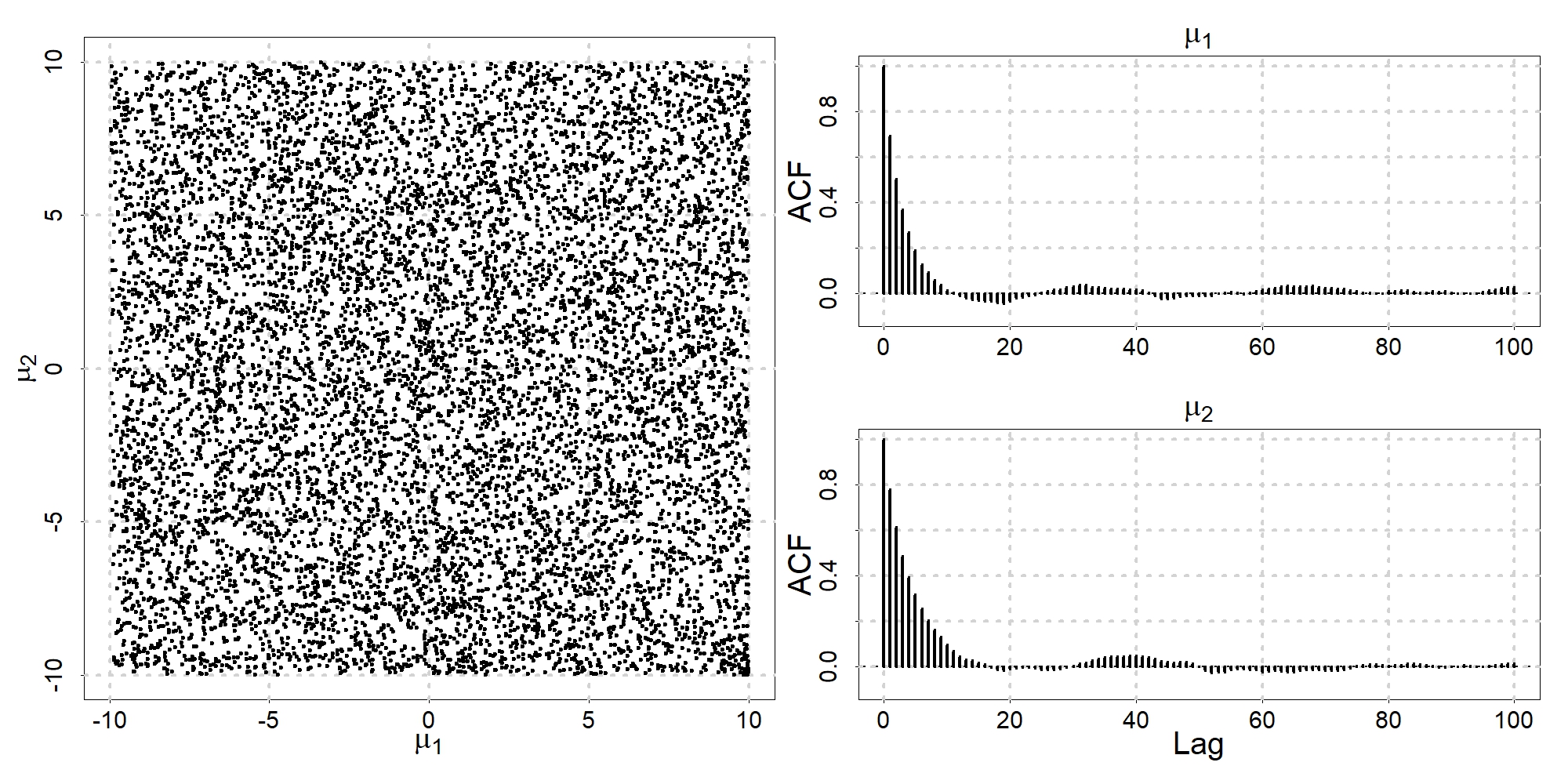}
		\caption{Simulation results for the $k=1{,}000$ case. The left and right panels depict the scatterplot of $(\mu_1, \mu_2)$ and autocorrelation functions of $\mu_1$ and $\mu_2$, respectively. The warmup iterations are not included.}
		\label{fig:unidentified_normal2}
	\end{figure}
    
    As \textcite{bacchiocchi2025} point out, besides flat posterior regions, a related but distinct challenge can arise when the identified set is multimodal because separated regions of posterior mass are difficult for MCMC algorithms (including HMC) to explore effectively. However as noted by \textcite{kitagawa2025}, such local identification arises only in the case of non-homogenous restrictions, that is, when a parameter is set to a non-zero value. As we preclude such restrictions, multimodality does not arise, and hence the sampler does not need to traverse between distinct modes. This absence of multimodality follows from the fact that the likelihood depends on $B$ only through $BB'$ and is therefore invariant under orthogonal rotations, implying a continuum of observationally equivalent factorizations rather than isolated peaks.\footnote{In this paper, we impose the normalization $\det(B)>0$ and restrict attention to a single normalized representation of the structural parameters to eliminate any remaining reflection symmetry. This normalization does not restrict the identified set beyond our existing identifying assumptions. To see this, consider the case where at most $N-1$ columns of $B$ are identified, so the last column is left unrestricted. Let $B_1$ satisfy the identifying restrictions and have $\det(B_1)<0$. Define $R=\mathrm{diag}(1,\ldots,1,-1)$ and $B_2 = B_1R$, i.e., $B_2$ equals $B_1$ with the sign of its last column flipped. Then $\det(B_2)=\det(B_1)\det(R)=-\det(B_1)>0$, while the first $N-1$ columns of $B_2$ and $B_1$ are equal. Hence, $B_2$ satisfies the same identifying restrictions and has the same economic interpretation. Moreover, the map $B\mapsto BR$ is bijective with inverse $B\mapsto BR^{-1}=BR$, which guarantees that there always exists exactly one such $B_2$.} As the preceding example indicates, the HMC sampler is well suited for exploring such distributions.

    Even though the NUTS has several advantages compared to commonly used MCMC algorithms, it is important to monitor the convergence of the MCMC chain using diagnostic checks. To that end, we recommend the tools of \textcite{vehtari2021}, the split-$\hat{R}$, and the bulk and tail effective sample size (ESS) statistics. We monitor these univariate statistics for all parameters because, as noted by \textcite{brooks1998}, spurious results may be obtained by monitoring only a subset of the parameters.
    
    The $\hat{R}$ measures chain mixing by comparing within- and between-chain variances. For a well-mixed chain, it approaches unity as the number of draws increases. When running a single chain (as is typical  in the SVAR literature), $\hat{R}$ is calculated by splitting the chain into two halves. The idea of the bulk and tail ESS, in turn, is to measure sampling efficiency by adjusting the number of iterations for the autocorrelation in the sample. The bulk ESS estimates the number of effectively independent samples in the MCMC output for estimating the mean or the median of the posterior, while the tail ESS measures sampling efficiency in the tails of the distribution. \textcite{vehtari2021} recommend accepting an MCMC sample only if $\hat{R} < 1.01$ and the ESS is at least 100 for a single chain. These thresholds ensure adequate mixing and a sufficiently low Monte Carlo standard error for estimation of the quantities of interest. 
    
    The univariate ESS has a multivariate counterpart \parencite{vats2019} that also accounts for correlations between the parameters when evaluating sample quality. Following \textcite{arias2025elliptical}, we report this statistic to assess the performance of the sampler for the IRF parameters of interest in Section \ref{sec:empirical}. However, we rely on $\hat{R}$ as the primary convergence diagnostic because it directly assesses chain mixing and can detect non-convergence even when the ESS is large. In contrast, the multivariate ESS is primarily a measure of sampling efficiency and may fail to indicate lack of convergence when the chain has not fully explored the target distribution.

    \subsection{Efficient Evaluation of the N-Inverse Wishart Density}\label{sec:efficient_evaluation}
	
	When we use a Hamiltonian Monte Carlo (HMC) algorithm to estimate the parameters, we need to repeatedly evaluate the posterior density of the structural parameters. As is evident from equation \eqref{eq:auxiliary_density}, and given the prior we have assumed, this entails repeated evaluation of the normal-inverse Wishart density
    	\begin{align}
		NIW_{(\tilde{\nu}, \tilde{\Omega}, \tilde{\Phi}, \tilde{S})}(\Sigma, \tA) \propto |\Sigma|^{-\frac{\tilde{\nu}+N+1}{2}} e^{-\frac{1}{2}\tr(\tilde{S}\Sigma^{-1})} |\Sigma|^{-\frac{Np+1}{2}} e^{ -\frac{1}{2} \vec(\tilde{A}-\tilde{\Phi})' (\Sigma\otimes\tilde{\Omega})^{-1} \vec(\tilde{A}-\tilde{\Phi}) }. \label{eq:niw_density}
	\end{align}

    We contribute to the literature by presenting a fast evaluation method for the N-IW density. Although the N-IW SVAR model is ubiquitous in the literature, computation of its density function is not routinely performed, as there are methods to generate N-IW random variables using direct sampling. Our method makes extensive use of solvers for triangular systems of equations, and the Cholesky factors of the elements of the covariance matrix of the N-IW distribution. Thus, it has similarities with the work of \textcite{carriero2016, chan2021}, who use the said Cholesky factors to efficiently generate normally distributed random variables with a Kronecker covariance structure.

	Evaluation of $|\Sigma|$ is  standard and uses the diagonal elements of the Cholesky factor of $\Sigma$. We obtain the Cholesky factor $L_\Sigma$ of $\Sigma$ by a Cholesky decomposition of $\Sigma = BB'$; another option would be to apply the $LQ$ decomposition\footnote{Software packages do not always provide a direct implementation of the $LQ$ decomposition. In that case, we can apply the QR decomposition to $B' = \tilde{Q}R$ and set $L = R'$ and $Q = \tilde{Q}'$. Another aspect to note is that the QR-decomposition for a square invertible matrix is unique only up to the signs of the diagonal of $R$. We assume the convention that these are positive, which yields a unique decomposition.} to the matrix $B$ and take $L$ to be the Cholesky factor of $\Sigma$. In our experience, the former is faster, but the latter provides the matrix $Q$ that might be of interest. Furthermore, as $|\det(B)| = \sqrt{\det(\Sigma)}$, there is no need to evaluate the determinant of $B$ separately to implement the density \eqref{eq:auxiliary_density}.
	
	The trace term in the first exponential on the right-hand side of \eqref{eq:niw_density} can be evaluated using the Cholesky decompositions $\tilde{S} = L_{\tilde{S}} L_{\tilde{S}}'$ and $\Sigma = L_\Sigma L_\Sigma'$. The value of $\tilde{S}$ does not change between the iterations of the sampler, so it is pre-computed and cached. Using the cyclic property of trace, we can rearrange the matrices as
	\begin{align}
		-\frac{1}{2}\tr(\tilde{S}\Sigma^{-1}) &= -\frac{1}{2}\tr(L_{\tilde{S}} L_{\tilde{S}}'L_\Sigma'^{-1} L_\Sigma^{-1}) \nonumber \\
		&= -\frac{1}{2}\tr(L_\Sigma^{-1} L_{\tilde{S}} L_{\tilde{S}}'L_\Sigma'^{-1} ) \nonumber  \\
		&= -\frac{1}{2}\tr(L_\Sigma^{-1} L_{\tilde{S}} [L_\Sigma^{-1} L_{\tilde{S}}]' ) \nonumber \\
		&= -\frac{1}{2}\tr(ZZ'),
	\end{align}
	where $Z = L_\Sigma^{-1} L_{\tilde{S}} $. Matrix $Z$ can be solved efficiently using forward substitution from the equation $L_\Sigma Z = L_{\tilde{S}}$. Finally, by the properties of the trace of a matrix product,
	\begin{align}
		-\frac{1}{2}\tr(\tilde{S}\Sigma^{-1}) &= -\frac{1}{2}\sum_{i=1}^{N}\sum_{j=1}^{N} Z_{ij}^2,
	\end{align}
	where $Z_{ij}$ denotes the $(i,j)$ element of $Z$.
	
    The computationally most intensive part is the evaluation of the latter exponential term in \eqref{eq:niw_density}. As in the case of $\tilde{S}$, the Cholesky decomposition $\tilde{\Omega} = L_{\tilde{\Omega}} L_{\tilde{\Omega}}'$ is also pre-computed and cached. Using the properties of the Kronecker product, we get
	\begin{align}
		-\frac{1}{2} \vec(\tilde{A}-\tilde{\Phi})'& (\Sigma\otimes\tilde{\Omega})^{-1}  \vec(\tilde{A}-\tilde{\Phi}) \nonumber \\
		&= -\frac{1}{2} \vec(\tilde{A}-\tilde{\Phi})' \Big((L_\Sigma L_\Sigma') \otimes (L_{\tilde{\Omega}} L_{\tilde{\Omega}}')\Big)^{-1} \vec(\tilde{A}-\tilde{\Phi}) \nonumber \\
		&= -\frac{1}{2} \vec(\tilde{A}-\tilde{\Phi})' (L_\Sigma'^{-1} \otimes L_{\tilde{\Omega}}'^{-1})(L_\Sigma^{-1}  \otimes L_{\tilde{\Omega}}^{-1}) \vec(\tilde{A}-\tilde{\Phi}) \nonumber \\
		&= -\frac{1}{2} \left[  (L_\Sigma^{-1} \otimes L_{\tilde{\Omega}}^{-1}) \vec(\tilde{A}-\tilde{\Phi}) \right]' \left[  (L_\Sigma^{-1} \otimes L_{\tilde{\Omega}}^{-1}) \vec(\tilde{A}-\tilde{\Phi}) \right] \nonumber \\
		&= -\frac{1}{2} \vec\left(  L_{\tilde{\Omega}}^{-1} (\tilde{A}-\tilde{\Phi}) L_\Sigma'^{-1} \right)' \vec\left(  L_{\tilde{\Omega}}^{-1} (\tilde{A}-\tilde{\Phi}) L_\Sigma'^{-1} \right).
	\end{align}
	Let us denote $W =  L_{\tilde{\Omega}}^{-1} (\tilde{A}-\tilde{\Phi}) L_\Sigma'^{-1}$. We can now solve for $W$ in two steps using forward substitution. First, solve $C$ from $CL_\Sigma' = (\tilde{A}-\tilde{\Phi})$. Next, solve $W$ from $L_{\tilde{\Omega}} W = C$. Finally, we get
	\begin{align}
		-\frac{1}{2} \vec(\tilde{A}-\tilde{\Phi})' (\Sigma\otimes\tilde{\Omega})^{-1} \vec(\tilde{A}-\tilde{\Phi}) &= -\frac{1}{2} \sum_{i=1}^{Np+1}\sum_{j=1}^{N} W_{ij}^2,
	\end{align}
	where $W_{ij}$ denotes the $(i,j)$ element of $W$. This calculation is substantially faster than explicitly computing the $N(Np+1) \times N(Np+1)$ matrix $\Sigma\otimes\tilde{\Omega}$ and working with it.

	\section{Empirical Applications} \label{sec:empirical}

    This section revisits two empirical applications to showcase the benefits of our approach: In Section \ref{sec:oilmarket}, we identify three shocks in the small-scale global oil market model of \textcite{kilian2014}, and in Section \ref{sec:35variable}, we consider the large-scale U.S. economy model recently analyzed by \textcite{chan2025}. In each case, the empirical findings coincide with prior results based on the same identifying restrictions, but we highlight the advantages of our method over existing alternatives. Also \textcite{arias2025elliptical} consider these applications to illustrate the performance of their Gibbs sampler, which facilitates straightforward comparison. As their approach represents the current state of the art in inference under identifying sign restrictions, it provides a natural benchmark for our analysis. In both applications, our method achieves convergence approximately two to four times faster than theirs, as measured by the time required for the $\hat{R}$ statistic to fall below the conventional threshold. The obtained samples also exhibit lower autocorrelations, as reflected in greater bulk and tail ESS in both cases. In the large-scale model, our approach is also significantly faster than the algorithm of \textcite{chan2025}, although they use the asymmetric conjugate prior \parencite[see][]{chan2022} rather than the normal-inverse Wishart prior. In Section \ref{sec:35variable}, we also show that our approach accommodates exclusion restrictions at no additional computational cost.
	
	\subsection{Small Oil Market Model} \label{sec:oilmarket}
    
	To illustrate how our approach handles contemporaneous and dynamic sign restrictions, as well as elasticity restrictions, we estimate the oil market model of \textcite{kilian2014} on monthly data from 1973:M2 to 2009:M8.\footnote{The data are obtained from the replication package for \textcite{kilian2014} in the Journal of Applied Econometrics Data Archive.} It contains four variables, namely, the change in global oil production, a measure of global real activity, the real price of oil, and the changes in global oil inventories, as well as an intercept, and monthly dummy variables. Following \textcite{kilian2014}, the lag length is set at 24. The prior is flat and proportional to $|\Sigma|^{-(N+1)/2}$, which results in a conjugate normal-inverse Wishart posterior for the reduced-form parameters. Following the identification strategy of \citeauthor{kilian2014}, we identify the flow supply, flow demand, and speculative demand shocks. In particular, we impose the sign and elasticity restrictions on the impact responses, as well as the dynamic sign restrictions summarized in Table \ref{tab:kilian_restrictions}. These restrictions result in a tightly identified model, and thus random sampling of the orthogonal rotation matrix $ Q $ would result in extremely few accepted posterior draws, as discussed in Section \ref{sec:id_restrictions}.
	
	\begin{table}[htbp]
    \centering
            \begin{minipage}{0.8\textwidth}
		\footnotesize
		\caption{Sign restrictions in the VAR model}
		\label{tab:kilian_restrictions}
                \resizebox{\linewidth}{!}{%
		\begin{tabular}{lccc}
			\toprule
			& Flow supply shock & Flow demand shock & Speculative demand shock \\
			\midrule
			Oil production     & --1 & +1 & +1 \\
			Real activity      & --1 & +1 & --1 \\
			Real price of oil  & +1  & +1 & +1 \\
			Oil inventories        &  0  &  0  & +1 \\
			\midrule
			\multicolumn{4}{c}{Dynamic sign restrictions (horizons 1--12)} \\
			\midrule
			Real activity & --1 & 0 & 0 \\
			Real price of oil & +1 & 0 & 0 \\
			\midrule
			\multicolumn{4}{c}{Elasticity bounds} \\
			\midrule
			Price elasticity of oil supply & & $(0,0.025)$ & $(0,0.025)$\\
			\bottomrule
		\end{tabular}
        		}
        \scriptsize{Negative and positive sign restrictions are indicated by $-1$ and $+1$, respectively, while 0 indicates no restriction.}
        \end{minipage}
	\end{table}
	
	Our approach solves the issue of tight identification by transforming the constrained structural parameters into unconstrained auxiliary parameters. Specifically, given that we have dynamic restrictions up to the 12th IRF horizon, we work with the constrained structural parameters
	\begin{align}
		\Pi_{12} = \left( \tilde{c}, \tilde{D}, B, \Psi_1,\dots, \Psi_{12}, \tA_{13},\dots \tA_{24}\right),
	\end{align}
	where $ \tilde{D} $ is a $4 \times 11$ matrix of dummy coefficients, with each column corresponding to one dummy variable. These parameters are transformed into an unconstrained vector $\theta$ using the transformations described in section \ref{sec:enforcing_restrictions}, specifically the exponential and inverse logit transformations on $B$ and $\Psi_{1},\dots, \Psi_{12}$. We  run 2{,}000 warmup iterations, followed by 100,000 sampling iterations.

    Figure \ref{kilian2014_irf} depicts the posterior IRFs. As expected, the results are similar to those reported by \textcite{kilian2014, arias2025elliptical}. In particular, following a negative flow supply shock, oil production, global economic activity and oil inventories decline, while the real price of oil increases. A positive flow demand shock leads to a persistent increase in global economic activity and the real price of oil. It also induces a positive response in oil production, which peaks after about one year and then returns to its pre-shock level. Finally, following a positive speculative demand shock, the real price of oil and inventories increase, while global economic activity and oil production experience a slight but persistent decline.
	
	\begin{figure}[htbp]
		\centering
		\includegraphics[width=0.8\textwidth]{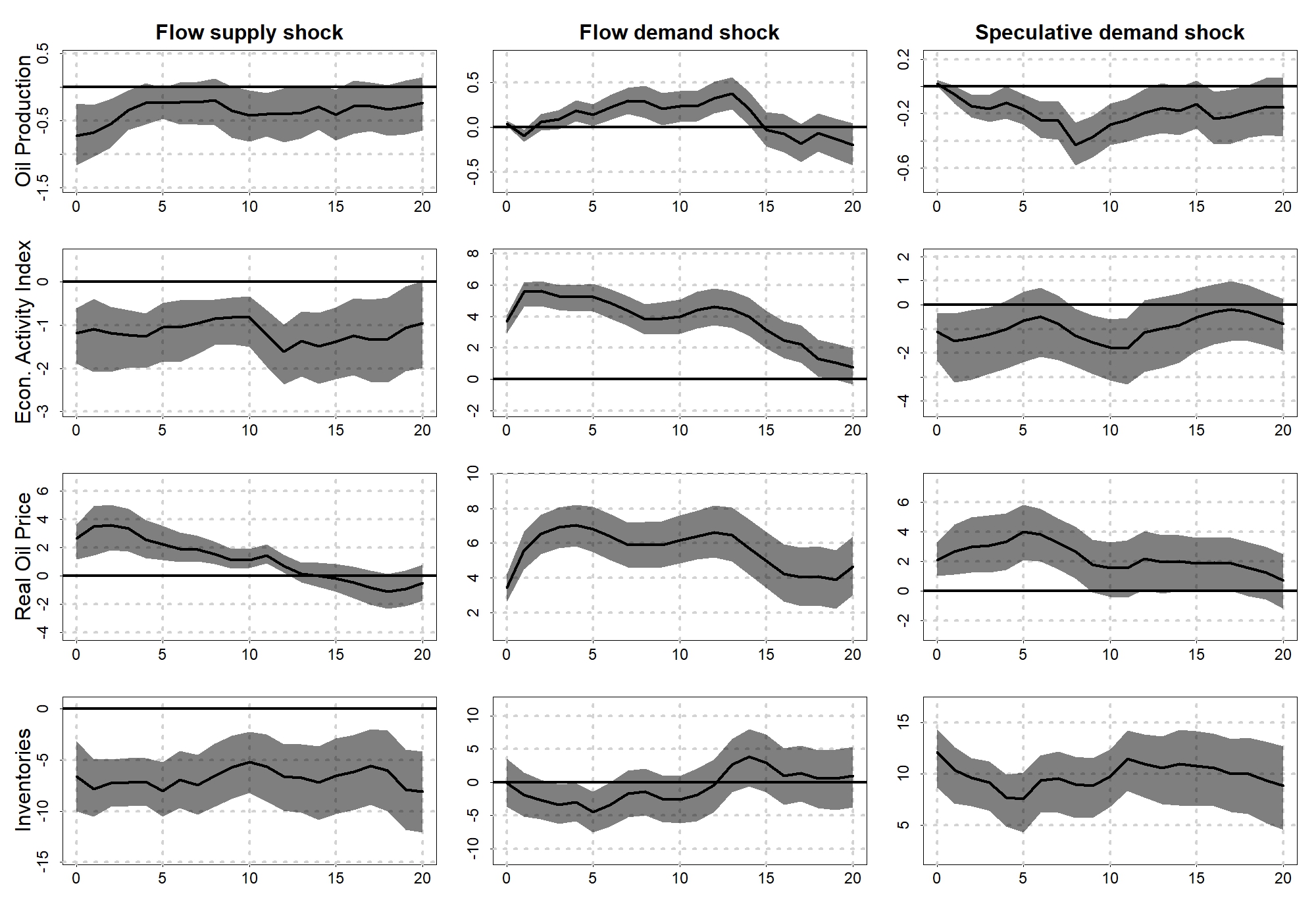}
		\caption{Impulse responses. The solid lines indicate posterior medians, and the shaded areas are the 68\% pointwise posterior probability regions.}
		\label{kilian2014_irf}
	\end{figure}
    
    We now turn to the properties of our sampler. Figure \ref{kilian2014_traceplot} presents trace plots of the unnormalized posterior log-density and a parameter ($B_{11}$) from the impact matrix $B$. For clarity, only the first 4{,}000 iterations, including warmup, are shown. As can be seen, both the log-density and the parameter are initiated near their typical level. This indicates that our method for selecting initial values by direct sampling works well and initializes the chain near the typical set. The traceplots do not exhibit problematic behavior such as nonstationarity or high autocorrelation after warmup. Furthermore, the variance of $B_{11}$ is lower during the early warmup phase than during sampling, which reflects the dynamic adaptation of the step size and inverse metric. This adaptation leads to more efficient sampling after warmup. The remaining structural parameters exhibit similar behavior and their traceplots are therefore not shown.
	
	\begin{figure}[htbp]
		\centering
		\includegraphics[width=0.8\textwidth]{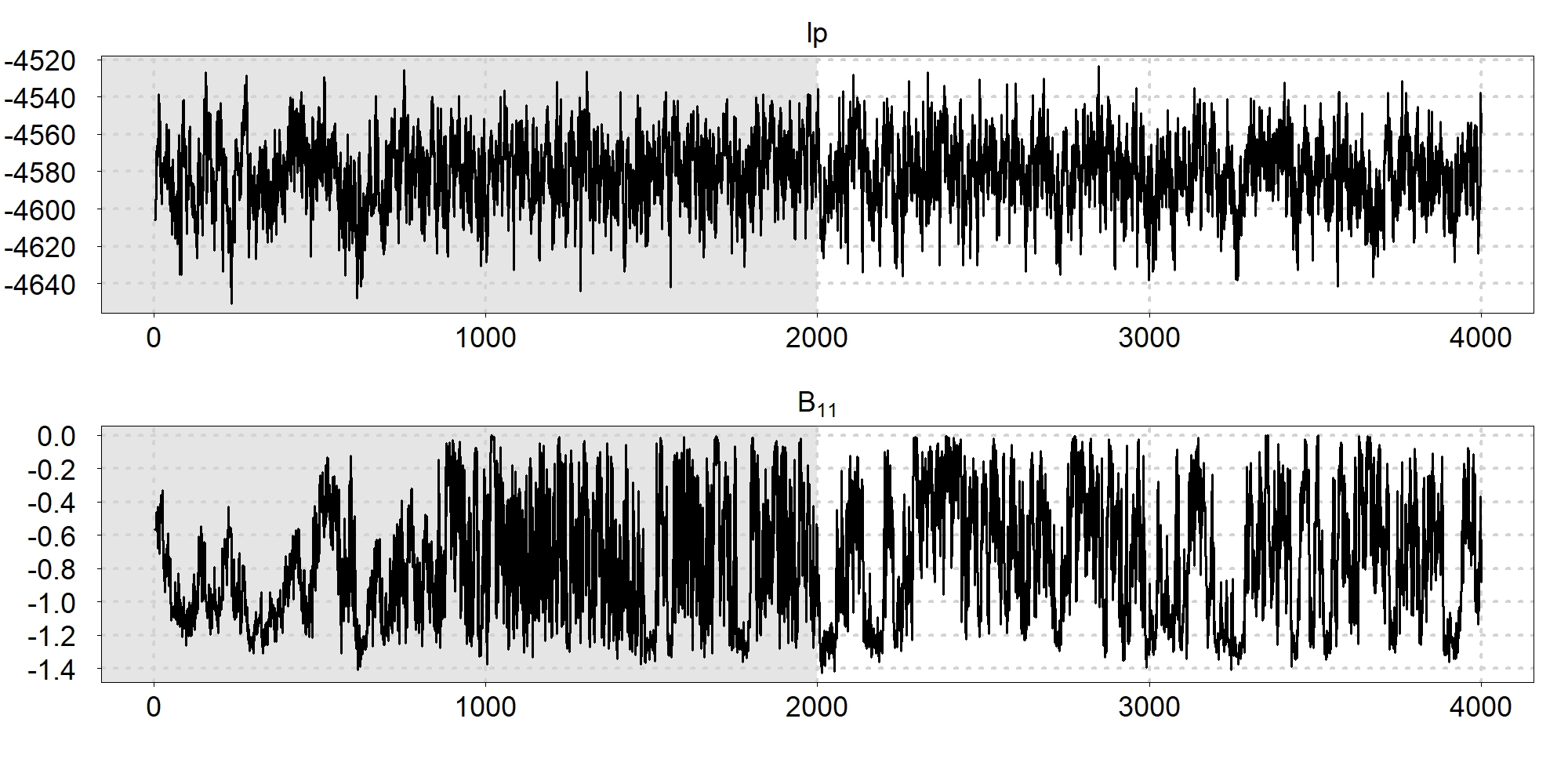}
		\caption{The traceplot of the unnormalized posterior log-density (top) and the impact effect of oil production to a flow supply shock (bottom) in the oil market model. The shaded area indicates warmup iterations.}
		\label{kilian2014_traceplot}
	\end{figure}
	
    To compare our approach with the elliptical slice sampler of \textcite{arias2025elliptical}, we implemented their sampler in Matlab\footnote{We thank Jonas Arias for generously sharing their replication package, which greatly facilitated our work. Building on this foundation, we implemented the elliptical slice sampler in R and Matlab, with the latter yielding a substantially faster implementation. We additionally experimented with implementing the large SVAR NUTS sampler of the next section in Matlab, which ran approximately one-third faster than the Stan version. Nonetheless, we base all reported results on Stan because Stan’s mature and widely tested implementation substantially reduces the scope for implementation errors and, through its automatic differentiation, obviates the need to derive and code gradients by hand, providing important practical and reliability advantages.} and, following their lead, save every tenth draw from the resulting chain. We run this sampler for a total of four million iterations. To assess convergence, we use the $\Rhat$ statistic (see Section \ref{sec:algorithm}), calculated from a single chain by considering samples of increasing length, with warmup iterations included.\footnote{Because there is no clear way to ensure that the discarded warmup samples are comparable across samplers, we include the warmup iterations in the analysis.}  As recommended by \textcite{vehtari2021}, we interpret values of $\Rhat$ above 1.01 as an indication of potential convergence issues.

    Figure \ref{fig:rhats_kilian} plots the maximum values of the $\Rhat$ statistic for samples obtained using both samplers. The reported sample size is increased in increments of 1,000 iterations. For the elliptical slice sampler, the maximum $\Rhat$ is the maximum over the elements of $\Sigma$ and $Q$ , whereas for the NUTS it is the maximum over the elements of $B$.\footnote{In our experience, the unidentified parameters $B$ and $Q$ mix the most slowly. To reduce computational cost, we therefore do not compute $\Rhat$ for all parameters.} As can be seen in the figure, $\Rhat$ converges in fewer iterations when using the NUTS. For the elliptical slice sampler, $\Rhat$ first falls below the threshold after approximately 27 thousand iterations (corresponding to roughly 13 minutes), but requires about 130 thousand saved draws (64 minutes) to remain consistently below the threshold of 1.01. In contrast, the NUTS reaches the 1.01 threshold after roughly nine thousand iterations, corresponding to a runtime of 28 minutes.

    \begin{figure}[htbp]
		\centering
		\includegraphics[width=0.8\textwidth]{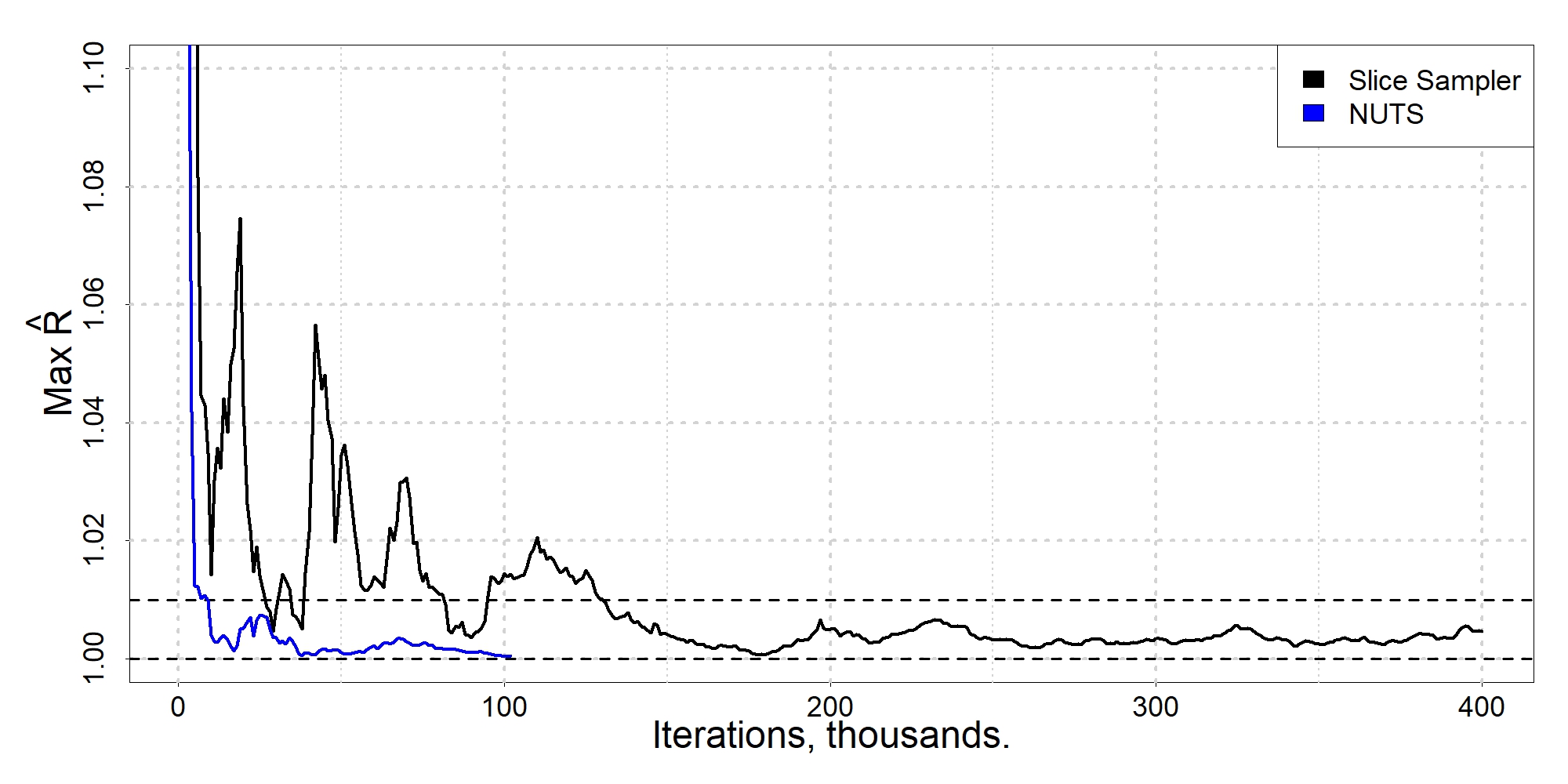}
		\caption{The maximum $\Rhat$ statistic for samples of different sizes from the elliptical slice sampler and NUTS. Horizontal dashed lines are located at 1.00 and 1.01. $\Rhat$ values are calculated for the elements of $B$ for the NUTS, and the elements of $\Sigma$ and $Q$ for the elliptical slice sampler. The vertical axis is restricted to $(1.0, 1.1)$ for clarity.}
		\label{fig:rhats_kilian}
	\end{figure}

    To investigate the sampling efficiency of the samplers, we consider the bulk and tail ESS measures of \textcite{vehtari2021}. In line with the $\Rhat$, they are calculated for all elements of $B$ (for the NUTS), and $\Sigma$ and $Q$ (for the slice sampler), again using chains of increasing length and including warmup iterations. Table \ref{tab:ess_kilian} presents the minimum ESS values per 1,000 iterations over these parameters for the NUTS and the elliptical slice sampler. As can be seen, the NUTS is more efficient in terms of the ESS per iteration under both measures. This is especially evident in the case of the bulk ESS, indicating that our approach yields more reliable estimates of posterior means for a given number of iterations. Importantly, given the lengths of the sampled chains, both methods reach ESS levels at which the Monte Carlo error is negligible for inference.

    \begin{table}[htbp]
        \centering
        \caption{Minimum ESS measures per 1000 iterations. The values are calculated for elements of $B$ for the NUTS, and $\Sigma$ and $Q$ for the elliptical slice sampler.}
        \begin{tabular}{rcc}
            \toprule
            & \multicolumn{2}{c}{ESS type} \\
            \cmidrule(lr){2-3}
            Sampler & Bulk & Tail \\
            \midrule
            NUTS & 40.53 & 48.24 \\
            Slice Sampler & 2.18 & 5.66 \\
            \bottomrule
        \end{tabular}
        \label{tab:ess_kilian}
    \end{table}

    \begin{figure}[htbp]
		\centering
		\includegraphics[width=0.8\textwidth]{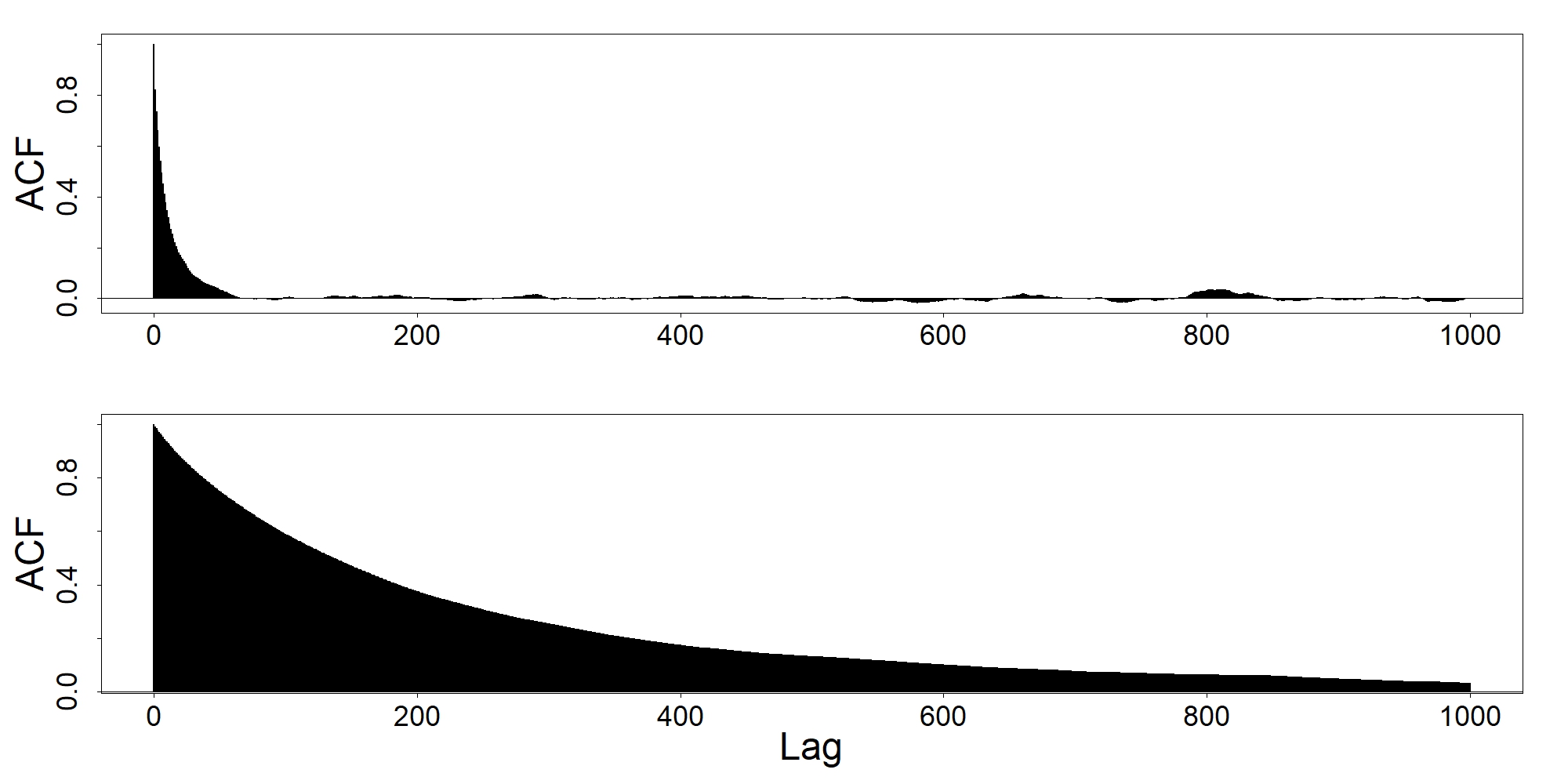}
		\caption{Autocorrelation plots of the sample of the $B_{11}$ obtained using the NUTS (top panel), and the $Q_{11}$ obtained using the elliptical slice sampler (bottom panel).}
		\label{fig:kilian_acf}
	\end{figure}

    The non-monotone behavior of the $\Rhat$ statistic and the lower ESS for the elliptical slice sampler are likely due to the high autocorrelation in the sampled chain. Figure \ref{fig:kilian_acf} presents the autocorrelation function for $B_{11}$ from the NUTS (top panel), and $Q_{11}$ from the elliptical slice sampler (bottom panel). The latter exhibits substantially higher autocorrelation, which reduces sampling efficiency. Similar patterns are observed for the other elements of the matrices \(B\) and \(Q\) in this application.

    We also assess convergence using the minimum multivariate ESS, as suggested by \textcite{vats2019}. To achieve a precision of 0.05 for the 95\% confidence region our approach requires 52 minutes of sampling, whereas the elliptical slice sampler requires 31 minutes. Precision of 0.05 implies that the Monte Carlo error is 5\% of the variability in target distribution. While this metric suggests that the elliptical slice sampler is more efficient in this setting, it should be interpreted with caution, as the multivariate ESS primarily reflects sampling efficiency and may fail to detect lack of convergence, as discussed in Section \ref{sec:algorithm}. \textcite{arias2025elliptical} report the efficiency in terms of time required to generate 1,000 multivariate ESS. The values are 4.76 and 5.17 minutes per 1,000 ESS for the NUTS (including warmup) and the elliptical slice sampler, respectively. Following \textcite{arias2025elliptical}, the multivariate ESS is calculated jointly for the first three columns of the impulse response parameters, corresponding to the identified shocks in this model.

	\subsection{Large Model of the U.S. Economy} \label{sec:35variable}
	
    As an example of a large SVAR model, we consider a 35-variable quarterly model of the U.S. economy with five lags and a constant. This application extends the work of \textcite{chan2025}, which in turn builds on \textcite{crump2021}. Following \textcite{arias2025elliptical}, we consider the sample period from 1977:Q4 to 2019:Q4.\footnote{We thank Christian Matthes for sharing their replication files.} The variables are listed in Table \ref{tab:var35}, which also gives the identifying restrictions. Following \textcite{arias2025elliptical}, both sign and ranking restrictions are imposed to identify ten structural shocks: demand, investment, financial, monetary, government spending, technology, labor supply, wage bargaining, oil price, and consumer sentiment shocks.
	
	Following \textcite{arias2025elliptical}, we place a Minnesota prior on the reduced-form parameters.\footnote{We thank Jonas Arias for sharing their prior with us.} As the model contains only contemporaneous restrictions, we estimate it using the non-centered specification (discussed brifely in Section \ref{sec:algorithm}), where instead of sampling the VAR coefficient matrices, we sample an auxiliary standard normal vector and transform it into the desired parameters $\tA$. Further details are provided in the Appendix \ref{app:non-centered}.

	\begin{table}[htbp]
		\centering
        \begin{minipage}{0.8\textwidth}
		\caption{Sign and ranking restrictions}
		\label{tab:var35}
         \resizebox{\linewidth}{!}{%
			\begin{tabular}{lrrrrrrrrrr}
				\toprule
				& \textbf{Dem} & \textbf{Inv} & \textbf{Fin} & \textbf{Mon} & \textbf{Gov} & \textbf{Tec} & \textbf{Lab} & \textbf{Wag} & \textbf{Oil} & \textbf{Con} \\
				\midrule
                \addlinespace
				\multicolumn{11}{l}{\textit{Sign restrictions}} \\
				\midrule
				GDP & +1 & +1 & +1 & -1 & +1 & +1 & +1 & +1 & +1 & +1 \\
				PCE & 0 & 0 & 0 & 0 & 0 & +1 & 0 & 0 & +1 & +1 \\
				Residential investment & 0 & 0 & 0 & 0 & 0 & 0 & 0 & 0 & 0 & +1 \\
				Nonresidential investment & 0 & +1 & 0 & 0 & 0 & +1 & 0 & 0 & +1 & +1 \\
				Exports & 0 & 0 & 0 & 0 & 0 & 0 & 0 & 0 & 0 & 0 \\
				Imports & 0 & 0 & 0 & 0 & 0 & 0 & 0 & 0 & 0 & 0 \\
				Government spending & 0 & 0 & 0 & 0 & +1 & 0 & 0 & 0 & 0 & 0 \\
				Fed. budget surplus/deficit & 0 & 0 & 0 & 0 & -1 & 0 & 0 & 0 & 0 & 0 \\
				Fed. tax receipts & 0 & 0 & 0 & 0 & +1 & 0 & 0 & 0 & 0 & 0 \\
				GDP deflator & +1 & +1 & +1 & -1 & +1 & -1 & -1 & -1 & -1 & +1 \\
				PCE index & +1 & +1 & +1 & -1 & +1 & -1 & -1 & -1 & -1 & +1 \\
				PCE index less F\symbol{38}E & +1 & +1 & +1 & -1 & +1 & -1 & -1 & -1 & -1 & +1 \\
				CPI index & +1 & +1 & +1 & -1 & +1 & -1 & -1 & -1 & -1 & +1 \\
				CPI index less F\symbol{38}E & +1 & +1 & +1 & -1 & +1 & -1 & -1 & -1 & -1 & +1 \\
				Hourly wage & 0 & 0 & 0 & 0 & 0 & +1 & -1 & -1 & +1 & 0 \\
				Labor productivity & 0 & 0 & 0 & 0 & 0 & +1 & 0 & 0 & +1 & 0 \\
				Utilization-adjusted TFP & 0 & 0 & 0 & 0 & 0 & +1 & 0 & 0 & +1 & 0 \\
				Employment & 0 & 0 & 0 & -1 & 0 & 0 & 0 & 0 & 0 & 0 \\
				Unemployment rate & -1 & -1 & -1 & +1 & -1 & -1 & +1 & -1 & +1 & +1 \\
				Industrial production index & +1 & +1 & +1 & -1 & 0 & 0 & 0 & 0 & 0 & 0 \\
				Capacity utilization & +1 & +1 & +1 & -1 & 0 & 0 & 0 & 0 & 0 & 0 \\
				Housing starts & 0 & 0 & 0 & 0 & 0 & 0 & 0 & 0 & 0 & 0 \\
				Disposable income & 0 & 0 & 0 & 0 & 0 & 0 & 0 & 0 & 0 & 0 \\
				Consumer sentiment & 0 & 0 & 0 & 0 & 0 & 0 & 0 & 0 & 0 & 0 \\
				Fed funds rate & +1 & +1 & +1 & +1 & +1 & 0 & 0 & 0 & 0 & 0 \\
				3-month T-bill rate & +1 & +1 & +1 & +1 & +1 & 0 & 0 & 0 & 0 & 0 \\
				2-year T-note rate & 0 & 0 & 0 & +1 & 0 & 0 & 0 & 0 & 0 & 0 \\
				5-year T-note rate & 0 & 0 & 0 & +1 & 0 & 0 & 0 & 0 & 0 & 0 \\
				10-year T-note rate & 0 & 0 & 0 & +1 & 0 & 0 & 0 & 0 & 0 & 0 \\
				Prime rate & +1 & +1 & +1 & +1 & +1 & 0 & 0 & 0 & 0 & 0 \\
				Aaa corporate bond yield & 0 & 0 & 0 & +1 & 0 & 0 & 0 & 0 & 0 & 0 \\
				Baa corporate bond yield & 0 & 0 & 0 & +1 & 0 & 0 & 0 & 0 & 0 & 0 \\
				Trade-weighted US index & 0 & 0 & 0 & 0 & 0 & 0 & 0 & 0 & 0 & 0 \\
				S\symbol{38}P 500 & 0 & -1 & +1 & -1 & 0 & 0 & 0 & 0 & 0 & +1 \\
				Spot oil price & 0 & 0 & 0 & 0 & 0 & 0 & 0 & 0 & -1 & 0 \\
				\midrule
				\addlinespace
				\multicolumn{11}{l}{\textit{Ranking restrictions}} \\
				\midrule
				Nonresidential investment/GDP & -1 & +1 & +1 & 0 & 0 & 0 & 0 & 0 & 0 & 0 \\
				Government spending/GDP & -1 & -1 & -1 & 0 & +1 & 0 & 0 & 0 & 0 & 0 \\
				\bottomrule
			\end{tabular}%
		}
        \scriptsize{Ten shocks are identified: demand (Dem), investment (Inv), financial (Fin), monetary (Mon), government spending (Gov), technology (Tec), labor supply (Lab), wage bargaining (Wag), oil price (Oil), and consumer sentiment (Con). In the top panel, a negative and positive impact effect of each shock is indicated by $-1$ and $+1$, respectively. In the bottom panel, $-1$ ($+1$) indicates that the shock has a smaller (greater) impact effect on the former variable in the leftmost column than on the GDP. Everywhere, 0 indicates no restriction.}
\end{minipage}
	\end{table}

	Figure \ref{fig:var35_irf1} depicts the impulse response functions of the real GDP, Fed funds rate, unemployment rate, PCE index, nonresidential investment, and real wage to a one standard deviation positive demand and investment shocks. The impulse responses are visually identical to those obtained using the elliptical slice sampler of \textcite{arias2025elliptical} under the same prior distribution and identifying restrictions, which suggests that our sampling algorithm correctly targets the same posterior distribution. In particular, the left panel of Figure \ref{fig:var35_irf1} shows that the demand shock causes a temporary increase in the real GDP, prices, and the federal funds rate, but reduces unemployment. The demand shock has no effect on nonresidential investment nor real wages. The right panel depicts impulse responses to the investment shock, which leads to a rise in nonresidential investment rises for three quarters, while the other responses are similar to those observed for the demand shock.
	
	\begin{figure}[htbp]
		\centering
		\includegraphics[width=0.48\textwidth]{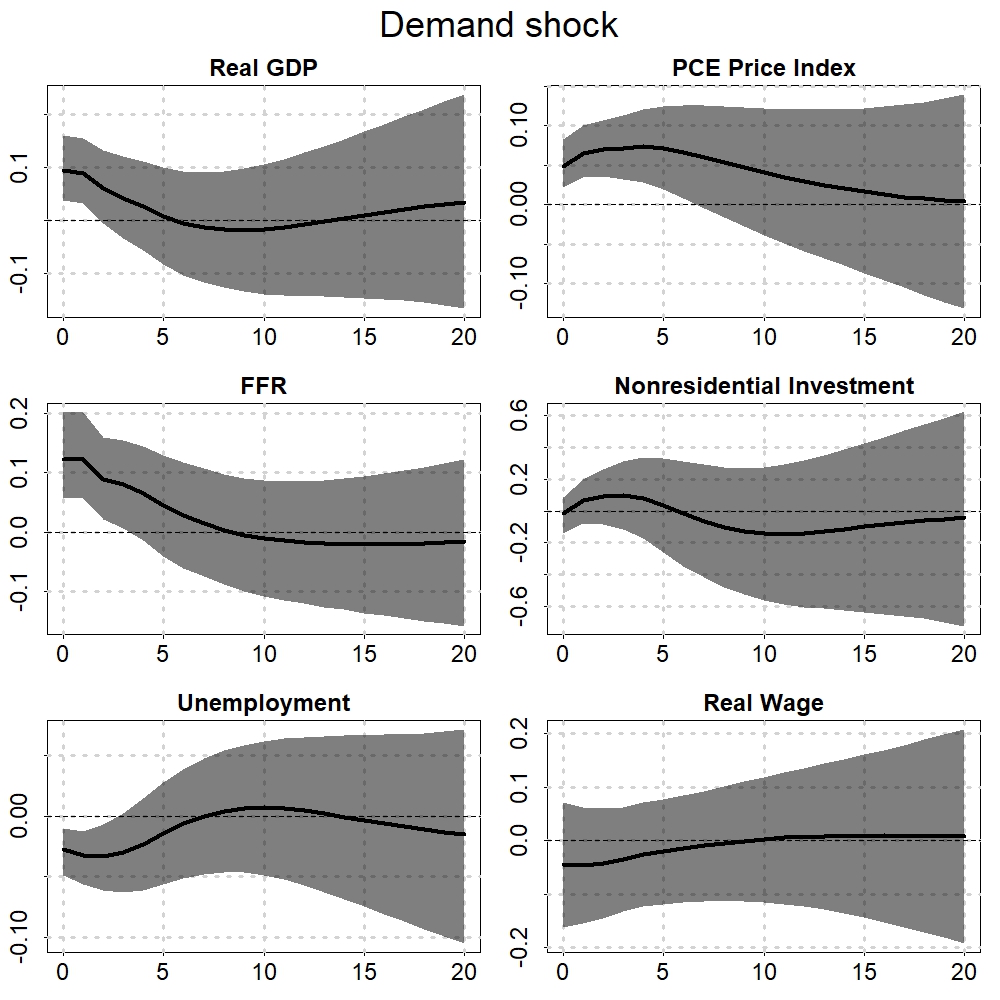}
		\hfill
		\includegraphics[width=0.48\textwidth]{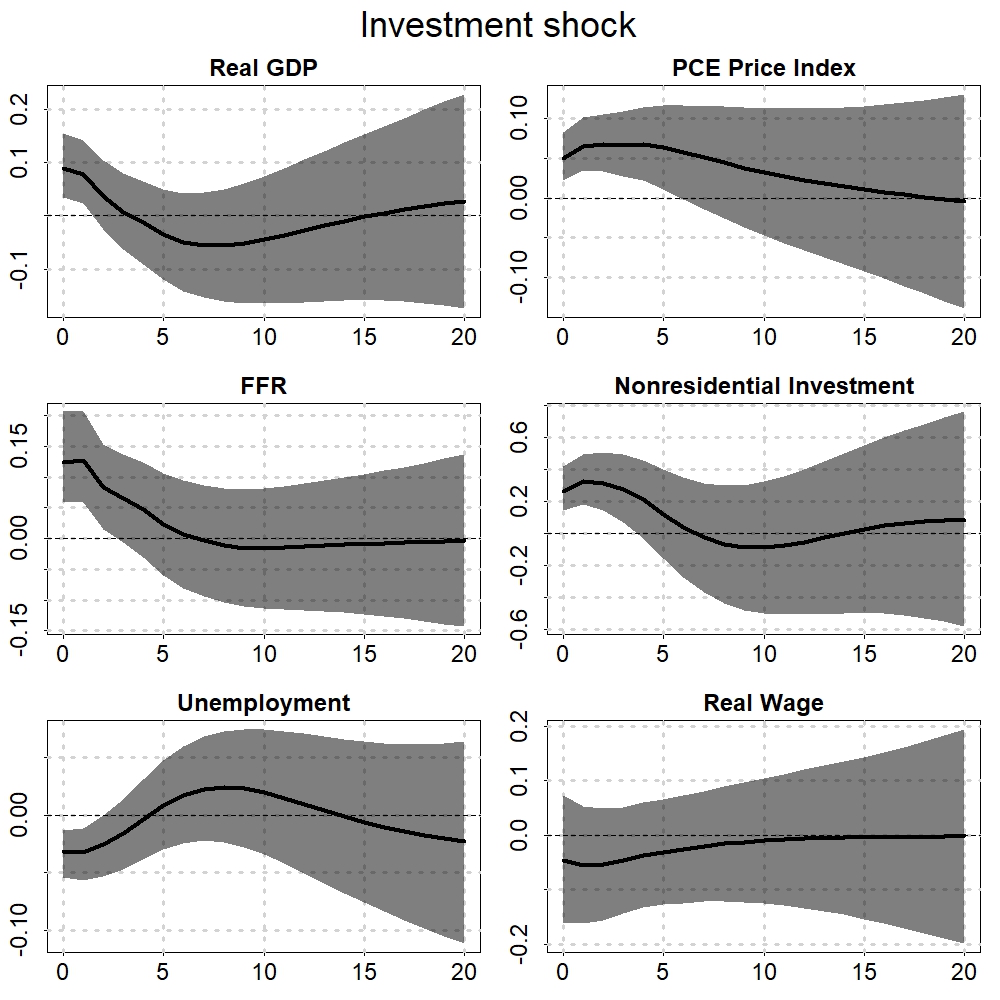}
		\caption{Impulse responses to one standard deviation demand (left panel) and investment (right panel) shocks. The solid lines indicate posterior medians, and the shaded areas are the 68\% pointwise posterior probability regions.}
		\label{fig:var35_irf1}
	\end{figure}
	
	Figure \ref{var35_lp} displays the traceplots of the unnormalized log-posterior and the  impact response of GDP to a demand shock. For clarity, only the first 4{,}000 iterations of the sampler are shown. The initial value lies within the typical set of the posterior distribution, and, therefore, no iterations are spent moving the chain into the typical set. Furthermore, the chains appear stationary and exhibit low autocorrelation, indicating a good performance of the sampler. Traceplots of other parameters are similar and are therefore not shown.	
	
	\begin{figure}[htbp]
		\centering
		\includegraphics[width=0.8\textwidth]{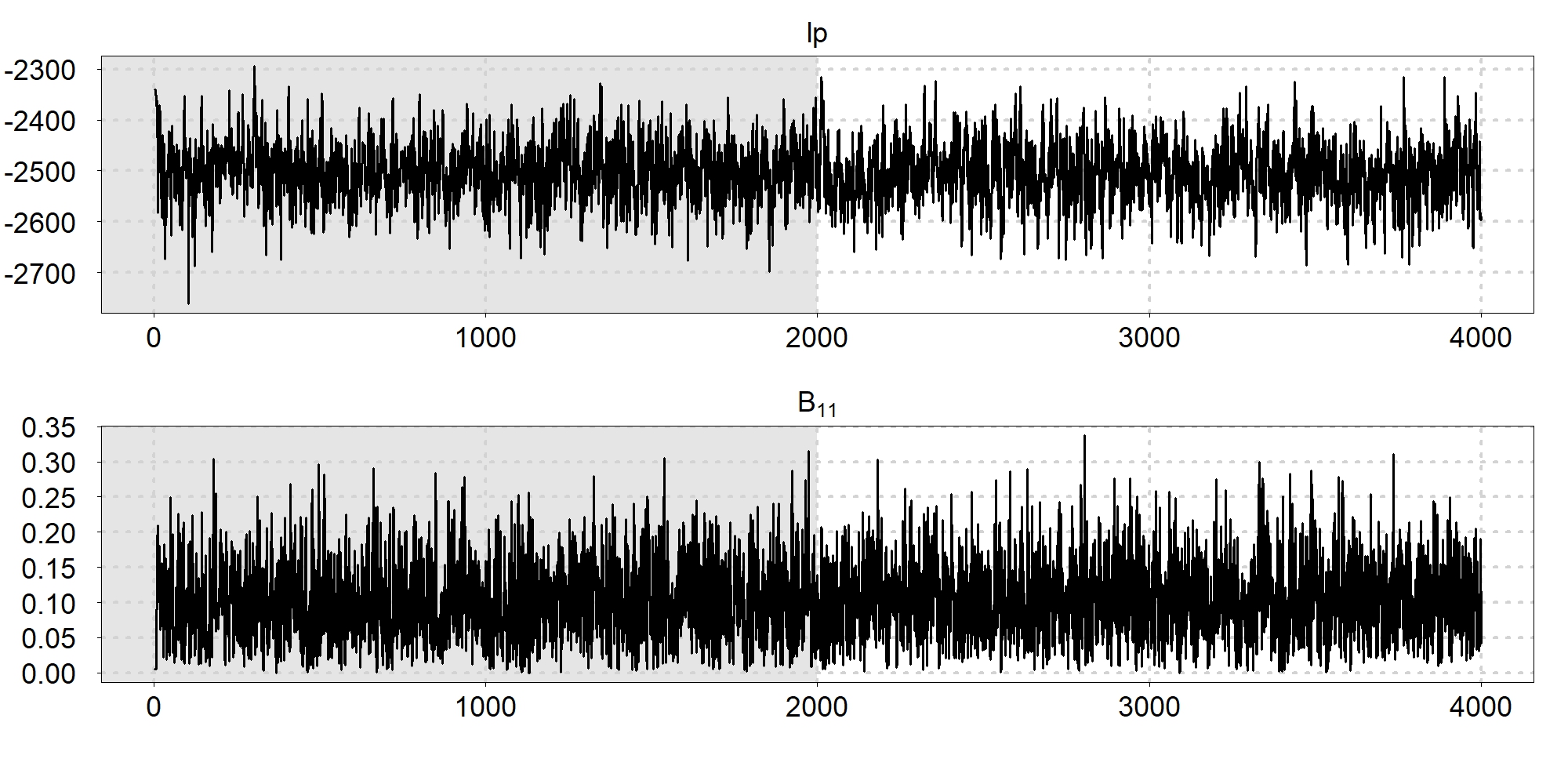}
		\caption{Traceplot of the unnormalized posterior log-density (top) and the impulse response of GDP to a demand shock (bottom). The shaded areas indicate warmup iterations. The horizontal axis denotes iterations.}
		\label{var35_lp}
	\end{figure}

    \begin{figure}[htbp]
		\centering
		\includegraphics[width=0.8\textwidth]{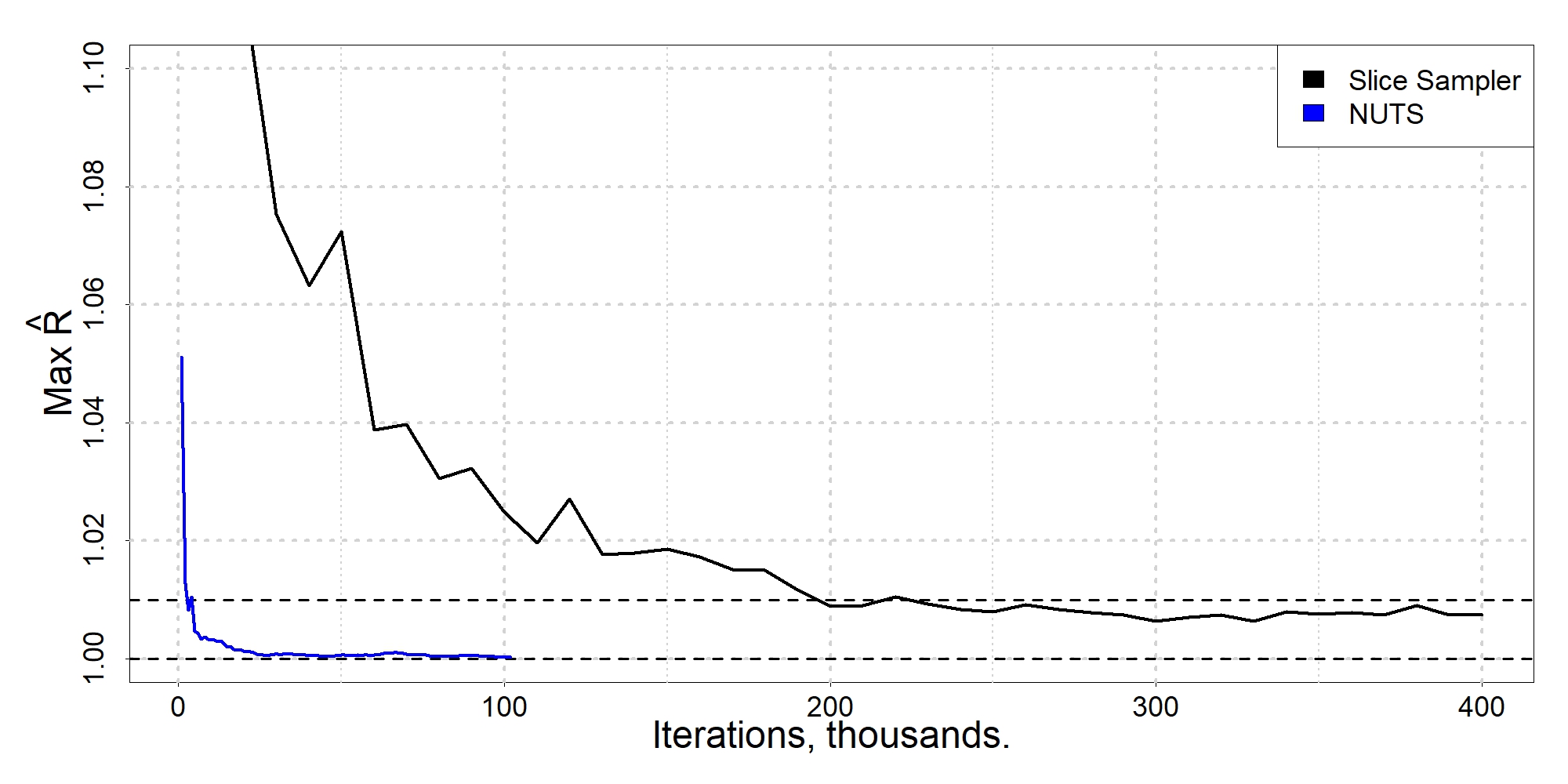}
        \caption{The maximum $\Rhat$ statistic for samples of different sizes from the elliptical slice sampler and NUTS. Horizontal dashed lines are located at 1.00 and 1.01. The values are calculated for the elements of $B$ for the NUTS, and the elements of $\Sigma$ and $Q$ for the elliptical slice sampler. The vertical axis is restricted to $(1.0, 1.1)$ for clarity.}
		\label{fig:rhats}
	\end{figure}

    \begin{figure}[htbp]
		\centering
		\includegraphics[width=0.8\textwidth]{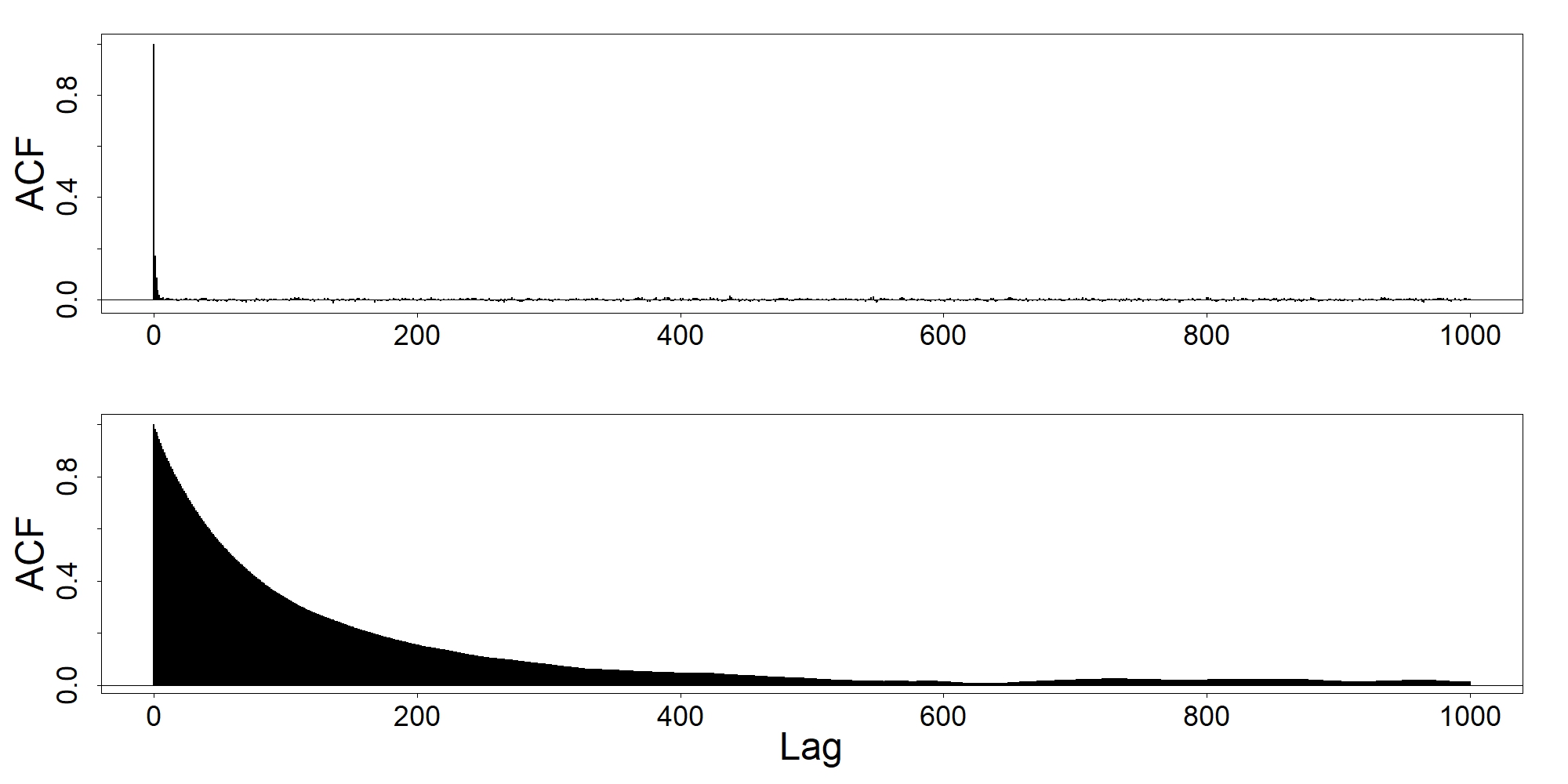}
		\caption{Autocorrelation plots of the sample of $B_{11}$ obtained using the NUTS (top panel), and $Q_{11}$ obtained using the elliptical slice sampler (bottom panel).}
		\label{fig:var35_acf}
	\end{figure}

    Figure \ref{fig:rhats} plots the maximum values of the $\Rhat$ statistic for samples obtained using both samplers. The $\Rhat$ is calculated for the elements of $\Sigma$ and $Q$ for the elliptical slice sampler, and for the elements of $B$ in the case of the NUTS. As in the previous application, we consider a threshold value of 1.01. As can be seen, $\Rhat$ converges much faster when sampling using the NUTS. The elliptical slice sampler requires roughly 200 thousand saved draws to reach the threshold of 1.01, while the NUTS only needs five thousand total iterations. Hence, the NUTS takes approximately six minutes to obtain an acceptable sample, whereas the elliptical slice sampler requires approximately 24 minutes. 
    
    The large number of iterations required by the elliptical slice sampler to converge suggests a high degree of autocorrelation also in this empirical application. Indeed, the chain generated by the elliptical slice sampler is highly autocorrelated, as Figure \ref{fig:var35_acf} shows. The autocorrelation coefficients for the $(1,1)$ element of the matrix $Q$ are positive for the first thousand lags, while the chain for the $(1,1)$ element of $B_{11}$ generated by the NUTS is uncorrelated beyond a few dozen lags.

    \begin{table}[htbp]
        \centering
        \caption{Minimum ESS measures per 1000 iterations. The values are calculated for elements of $B$ for the NUTS, and $\Sigma$ and $Q$ for the elliptical slice sampler.}
        \begin{tabular}{rcc}
            \toprule
            & \multicolumn{2}{c}{ESS type} \\
            \cmidrule(lr){2-3}
            Sampler & Bulk & Tail \\
            \midrule
            NUTS & 217.21 & 341.21 \\
            Slice Sampler & 2.24 & 5.11 \\
            \bottomrule
        \end{tabular}
        \label{tab:ess_var35}
    \end{table}

    In this large SVAR, the efficiency per iteration of our method in terms of bulk and tail ESS for individual parameters is significantly better than that of the elliptical slice sampler, as indicated by Table \ref{tab:ess_var35}. The minimum bulk ESS generated per 1000 iterations for the elements of $B$ is 217.21 when using the NUTS. In contrast, the minimum bulk ESS generated by the slice sampler for the elements of $\Sigma$ and $Q$ is 2.24. In the case of tail ESS, the corresponding values are 341.21 for the NUTS and 5.11 for the slice sampler. 

    As in Section \ref{sec:oilmarket}, we also assess convergence using the minimum multivariate ESS of \textcite{vats2019} (the multivariate ESS is calculated for the parameters in columns 1--10 corresponding to the identified shocks of interest). To achieve a precision of 0.05 for the 95\% confidence region, the NUTS requires roughly 7,000 iterations, whereas the elliptical slice sampler requires 124,000 iterations. These figures correspond to runtimes of 8 and 15 minutes, respectively. The time required to obtain 1{,}000 effective samples is 1.24 minutes for the NUTS, compared to 2.64 minutes for our implementation of the elliptical slice sampler. Our approach is faster than the algorithm of \textcite{chan2025}, although their results are not directly comparable to ours, as they use the asymmetric conjugate prior of \textcite{chan2022}, and only identify eight shocks instead of ten. They report a runtime of 140 minutes to obtain 1,000 draws, and as their draws are independent, this corresponds to an ESS of 1,000.
    
	\begin{figure}[htbp]
		\centering
		\includegraphics[width=0.48\textwidth]{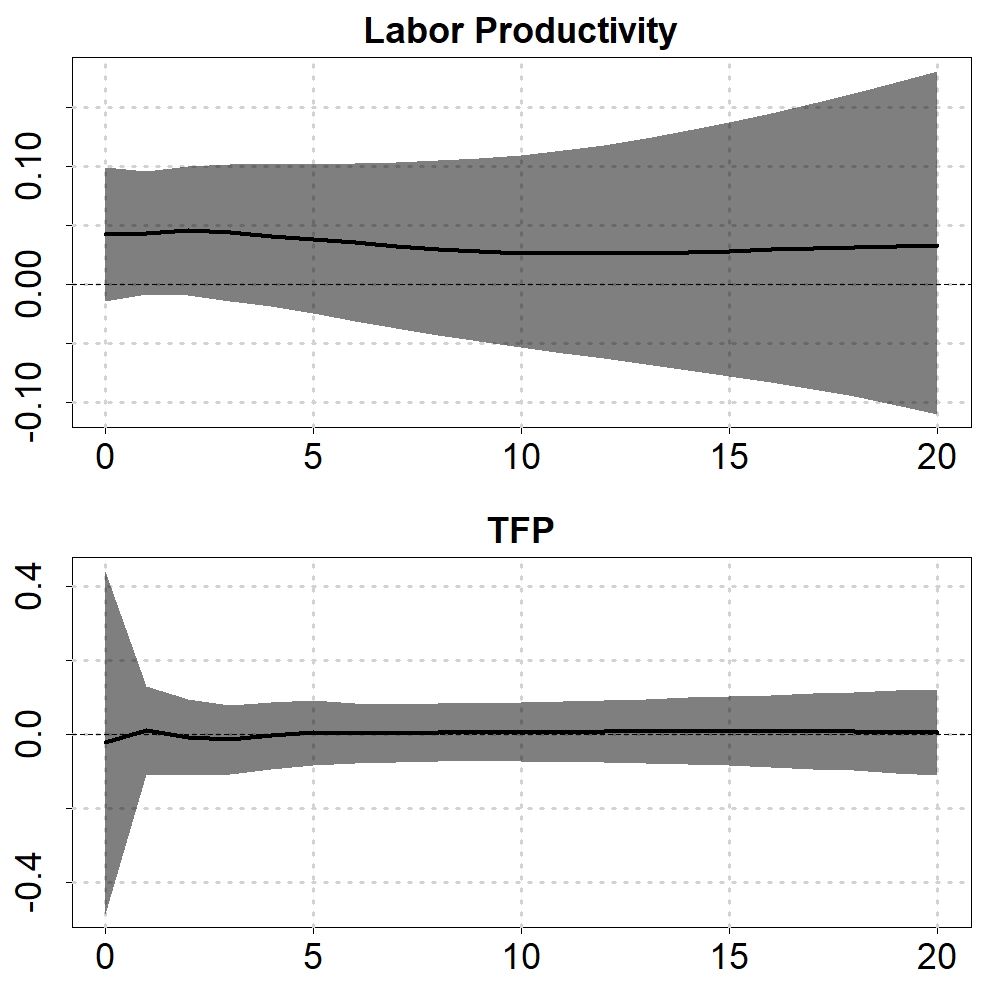}
		\hfill
		\includegraphics[width=0.48\textwidth]{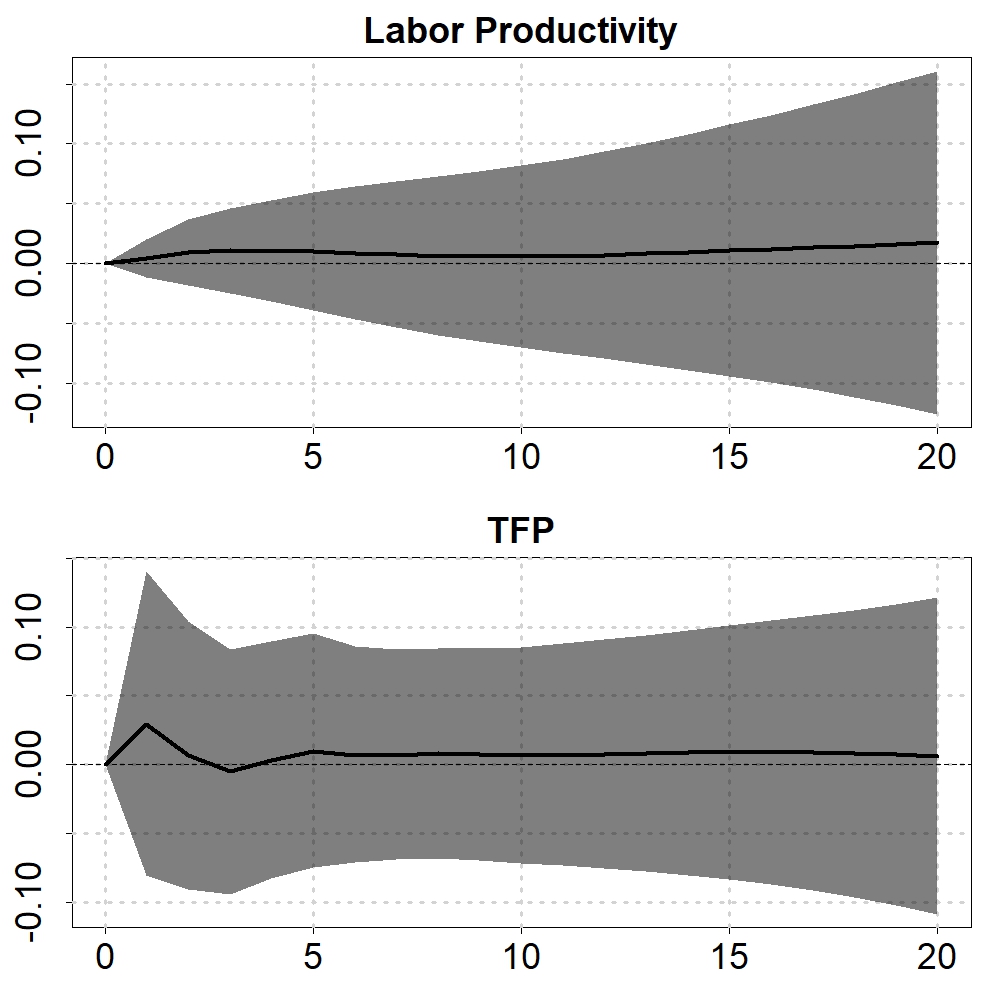}
		\caption{Impulse responses to a one standard deviation labor supply shock under the baseline identification (left panel) and with additional zero restrictions (right panel). The solid lines indicate the posterior medians, and the shaded areas are the 68\% pointwise posterior probability regions.}
		\label{fig:var35_irf_prod}
	\end{figure}
	
	To demonstrate that our approach can simultaneously incorporate both sign and zero restrictions, we augment the identification scheme with the restriction that labor productivity and utilization-adjusted TFP react on impact only to the technology and oil price shocks (assumed to have a positive impact effect). As discussed in Section \ref{sec:enforcing_restrictions}, introducing these zero restrictions does not require importance sampling to ensure that the sample is from the correct distribution. Instead, the restricted values are simply fixed in the density function and omitted from the parameter vector. Hence, estimation time is largely unchanged by the addition of these restrictions. The bulk and tail effective sample sizes per 1,000 iterations generated by the algorithm are 127.49 and 271.75, respectively, while the maximum $\Rhat$ crosses the 1.01 threshold after 5 thousand iterations, or 7 minutes. In terms of the multivariate ESS, the sampling efficiency comes at 1.72 minutes for 1{,}000 effective draws (as before, the effective draws are calculated for the impulse response parameters in columns 1--10 corresponding to the identified shocks of interest). Thus, sampling is less efficient than in the case of having no zero restrictions, but the convergence time in terms of the $\Rhat$ is essentially identical.

	In general, the impulse responses depicted in Figure \ref{fig:var35_irf1} remain intact. However, the responses to the labor supply shock shown in Figure \ref{fig:var35_irf_prod} change substantially when the zero restrictions are imposed. Under the baseline identification (left panel), the response of labor productivity has a substantial amount of posterior mass above zero, while the impact response of TFP is centered around zero but there is a lot of uncertainty about this estimate. In contrast, under the zero restrictions, the 68\% posterior credible intervals are much narrower at short horizons.

	\section{Conclusion}\label{sec:conclusion}

    Tight identifying inequality restrictions may render conventional accept-reject algorithms inefficient and in some cases infeasible for inference in SVAR models. In large-scale models that are becoming increasingly popular in empirical research, the high-dimensional parameter space and the large number of imposed restrictions can give rise to problems, even when individual restrictions are relatively weak, as their cumulative effect may still be substantial. In this paper, we propose an approach based on imposing inequality restrictions through a reparameterization of the structural model by continuously differentiable mappings. This approach accommodates various kinds of inequality restrictions, including shape and ranking restrictions on impulse responses at multiple horizons, as well as bounds on elasticities.  Moreover, in contrast to the approaches of \textcite{arias2025elliptical, chan2025, read2025}, our framework allows for a straightforward implementation of exclusion restrictions, with no significant increase in the speed of the sampler. As an additional advantage, the prior can be specified in a very flexible manner, as our approach does not rely on direct sampling procedures, albeit this usually entails an increase in computation time.

    While our framework is very flexible, it necessitates relying on MCMC methods because the posterior distribution of the parameters does not belong to a standard family, and we show that the NUTS algorithm, a variant of the HMC sampler, performs extremely effectively in exploring the posterior. However, each iteration of the NUTS requires repeated evaluation of the log-posterior and its gradient, which makes sampling computationally intensive. To address this, we show how the normal-inverse Wishart posterior density can be evaluated efficiently under the proposed reparameterization, facilitating inference even in high-dimensional parameter spaces.
    
    Recently, \textcite{kitagawa2025} have raised doubts about the effectiveness of MCMC algorithms in Bayesian analysis under set identification, but our simulation experiments indicate that the implementation in the Stan software package \parencite{stanreference} is effective in exploring high-dimensional, set-identified posteriors with flat directions and produces well-mixed chains even when the identified set is large. 

    Empirical illustrations with small-scale and large-scale SVAR models highlight the good performance of our approach. In particular, it exhibits clearly faster mixing than the Gibbs sampler of \textcite{arias2025elliptical}, while it turns out that the latter can generate Markov chains with extremely persistent serial correlation even at hundreds of lags, resulting in slow mixing and a very small effective sample size.

	\clearpage

	\printbibliography

	\clearpage

\begin{appendices}
\numberwithin{equation}{section} 

    \section{Proof of Proposition \ref{prop:structural_density}}\label{app:structural_density}
    In this appendix we prove Proposition \ref{prop:structural_density}, which gives us the unconstrained posterior of the structural parameters in terms of the normal-inverse Wishart posterior over the orthogonal reduced-form parameters. The proof builds a mapping from the structural parameters to the reduced-form parameters and derives the corresponding Jacobians for transforming the density function.
    
    We apply the change of variables in two steps, and for each step, we need to derive the absolute Jacobian determinant of the inverse transformation involved in order to characterize the induced posterior density over the structural parameters. First, we derive the Jacobian determinant $|J_1|$ of the transformation $g_1$ of the IRF matrices $\Psi_i$ into the corresponding reduced-form VAR matrices $\tilde{A}_i$. Then, we derive the the Jacobian determinant $|J_2|$, corresponding to the transformation $g_2$ of the structural impact matrix $B = \Psi_0$ into the reduced-form parameters $\Sigma$ and $Q$. The overall Jacobian determinant is then $|J_1||J_2|$.
	
	For the first inverse transform $g_1$, the IRF parameters satisfy the recursion
	\begin{align}
		\Psi_i &= \sum_{j=1}^{i} \tilde{A}_j\Psi_{i-j}, \ i\in\{1,\dots,k\},\nonumber
	\end{align}
	from which we obtain the inverse mapping
	\begin{align}
		\tilde{A}_i &= \left( \Psi_i - \sum_{j=1}^{i-1} \tilde{A}_j\Psi_{i-j} \right)\Psi_0^{-1},\ i\in\{1,\dots,k\}.\label{eq:var_recursion}
	\end{align}
	This inverse mapping is needed to calculate the Jacobian necessary for transforming the density function,
	\begin{align}
		|J_1| = \left| \frac{\partial \vec{(\tilde{c}, B,\tilde{A}_1,\dots, \tilde{A}_p)} }{\partial \vec(\tilde{c}, \Psi_0,\dots,\Psi_k, \tilde{A}_{k+1}, \dots, A_p)} \right|. \nonumber 
	\end{align}
	
	The construction of the Jacobian determinant turns out to be rather straightforward. Due to the relationship $B=\Psi_0$ and the fact that the formula in \eqref{eq:var_recursion} is recursive, it follows that the Jacobian is a block-triangular matrix. Thus its determinant is the product of the determinants of the diagonal blocks. The first two diagonal blocks of dimensions $N\times N$ and $N^2\times N^2$ correspond to the identity maps on $\tilde{c}$ and $B$, respectively, and are thus identity matrices. The same is true for the $N^2 \times N^2$ blocks in positions $k+1,\dots,p$, as they correspond to identity maps on matrices $\tilde{A}_i$. The remaining blocks on the diagonal correspond to the quantities
	\begin{align}
		\frac{\partial\vec(\tilde{A}_i)}{\partial\vec(\Psi_i)}, \nonumber
	\end{align}
	and as the only term containing $\Psi_i$ in the summation in \eqref{eq:var_recursion} is $\Psi_i\Psi_0^{-1}$, we have
	\begin{align}
		\left|\frac{\partial\vec(\tilde{A}_i)}{\partial\vec(\Psi_i)}\right| &= \left|\frac{\partial\vec(\Psi_i\Psi_0^{-1})}{\partial\vec(\Psi_i)}\right| = \left| \frac{\partial((\Psi_0')^{-1}\otimes I_N)\vec(\Psi_i)}{\partial\vec(\Psi_i)} \right| = |(\Psi_0')^{-1}\otimes I_N| = |\Psi_0|^{-N}. \nonumber
	\end{align}
	As there are $k$ such blocks on the diagonal, it follows that the Jacobian determinant $|J_1|$ is $|\Psi_0|^{-kN}=|B|^{-kN}$. Notice that if only impact responses are restricted, then $k=0$ and this determinant equals $1$.
	
	For the second inverse transformation $g_2$ we need the Jacobian $J_2$ of the mapping between the structural impact matrix $B$ and the matrices $\Sigma$ and $Q$. The relationship $\Sigma = B B'$ is not useful because the Jacobian of this transformation is not full rank. This follows from the fact that this transformation maps a full $N\times N$ matrix into the space of symmetric, positive definite matrices, which is of dimension $N(N+1)/2 < N^2$. Instead, we take the $LQ$-decomposition of $B=LQ$, where $L$ and $Q$ are a lower triangular and an orthogonal matrix, respectively, and set $\Sigma = LQQ'L' = LL'$. To ensure that the decomposition is unique, we assume that the diagonal of $L$ is positive. This is the well-known Cholesky parametrization of a symmetric positive-definite matrix. We thus avoid ending up with a singular Jacobian, and the Jacobian determinant for the Cholesky decomposition is $2^N\prod_{i=1}^{N}L_{ii}^{N-i+1}$ \parencite[pp.134]{stanreference}. The final ingredient for this transformation is the Jacobian of the $LQ$ decomposition itself, which is $\prod_{i=1}^{N}L_{ii}^{-(N-i)}$ \parencite[Theorem 3.1]{edelman1989}. Multiplication of this by the Jacobian of the Cholesky decomposition yields the full Jacobian determinant $|J_2|=2^N \prod_{i=1}^{N} L_{ii}$. We also recognize that $\prod_{i=1}^{N} L_{ii} = |B|$.
	
	Combining all these derivations together, we see that if the distribution over the parameters $(\Sigma, Q, \tilde{A})$ is defined by the density $f$, then the density over the structural parameters $\Pi_k$ is
	\begin{align*}
		p(\Pi_k) \propto f(\Sigma, Q, \tilde{A})|B|^{-kN+1}, 
	\end{align*}
	where we do not make explicit the dependence of the right-hand side parameters on the left-hand side parameters. The proposition follows by taking $f$ to be the normal-inverse Wishart posterior of the orthogonal-reduced form parameters.

	\section{Proof of Proposition \ref{prop:conditional_density}}\label{app:conditional_density}
	In this appendix, we provide a proof that the density of a random vector conditional on equality and set restrictions is proportional to the unrestricted density truncated to the constrained region. This is crucial for the efficiency of our sampler, as calculating normalizing terms could be extremely expensive.
	
	Suppose the random vector $\theta$ has a continuous density over $\R^{n_1 + n_2 + n_3}$, and partition it as $\theta = (\theta_1, \theta_2, \theta_3)$, with the dimension of $\theta_i$ being $n_i$. We use $f$ to denote a density function. Suppose $c\in\R^{n_2}$ is a vector of constants and $\int_{R} \int_{\R^{n_3}} f(\theta_1, c, \theta_3)\dif{\theta_3}\dif{\theta_1} > 0$, where $R$ is a set of strictly positive probability. We proceed in two steps. First, we condition $\theta$ on the event $\theta_1 \in R\subset\R^{n_1}$. In the second step, we condition the resulting density on $\theta_2  = c$.
	
	Conditioning on $\theta_1 \in R$ yields
	\begin{align*}
		f(\theta_1, \theta_2, \theta_3 \mid \theta_1 \in R) &= \frac{\mathds{1}_{\{\theta_1 \in R\}} f(\theta_1, \theta_2, \theta_3)}{P(\theta_1 \in R)} = \frac{\mathds{1}_{\{\theta_1 \in R\}} f(\theta_1, \theta_2, \theta_3)}{ \int_{R} \int_{\R^{n_2}} \int_{\R^{n_3}} f(\theta_1, \theta_2, \theta_3)\dif{\theta_3}\dif{\theta_2}\dif{\theta_1} }  \\
		&\propto \mathds{1}_{\{\theta_1 \in R\}} f(\theta_1, \theta_2, \theta_3) 
	\end{align*}

	By definition of conditional density, we can further condition this on $ \theta_2 = c $. to obtain
	\begin{align*}
		f(\theta_1, \theta_2, \theta_3 \mid \theta_1 \in R, \theta_2 = c) &= \frac{\mathds{1}_{\{\theta_1 \in R\}} f(\theta_1, c, \theta_3)}{\int_{R} \int_{\R^{n_3}} f(\theta_1, c, \theta_3)\dif{\theta_3}\dif{\theta_1}}  \\
		&\propto \mathds{1}_{\{\theta_1 \in R\}} f(\theta_1, c, \theta_3)
	\end{align*}

	\section{Estimation Algorithm}\label{app:algorithm}
	
	The transformations we have described in the text result in a density with no known direct sampling method. Instead, we use a Metropolis-style sampler, specifically a variant of Hamiltonian Monte Carlo (HMC) called the No-U-Turn Sampler (NUTS). The HMC algorithm was originally introduced by \textcite{duane1987} for simulating lattice field theory, while \textcite{neal1996} introduced this method to more general problems of statistical inference. In this appendix, we give a brief overview of the algorithm, following the notation of \textcite{stanreference}.
	
	In HMC, we augment the target density $p(\theta)$ with auxiliary momentum variables $\rho$ and draw from the join density $p(\theta, \rho) = p(\rho|\theta)p(\theta)$. In most applications the momentum variables are independent from the parameters, so that $ p(\theta, \rho) = p(\theta)p(\rho)$, and, furthermore, it is assumed that $\rho\sim N(0, M)$, where the covariance matrix $M$ is known as the ``metric'' of the algorithm.
	
	Given the momentum variables, we define the Hamiltonian as 
	\begin{align}
		H(\theta, \rho) &= -\log p(\theta) - \log p(\rho)  \label{hamiltonian} \\ 
		&= V(\theta) + T(\rho). \nonumber
	\end{align}
	The term $V(\theta)$ is the ``potential energy,'' while $T(\rho)$ is the ``kinetic energy.'' Taking the derivatives with respect to fictitious time gives Hamilton's equations
	\begin{align*}
		\frac{d\theta}{d t} &= \frac{\partial T}{\partial \rho} \\
		\frac{d\rho}{d t} &= \frac{\partial V}{\partial \theta}
	\end{align*}
	
	These two differential equations define how the Markov chain moves from one state to another. In practice, they do not have a closed-form solution except for the simplest problems, so the dynamics are approximated by a leapfrog integrator. At the beginning of each iteration, a random momentum vector $\rho$ is sampled, after which the equations are simulated forward in time by taking $L$ steps of length $\varepsilon$. Within each step, the momentum $\rho$ and position $\theta$ are updated in order as
	\begin{align*}
		\rho^{t+\varepsilon/2} &= \rho^t + (\varepsilon/2)\nabla_\theta\log p(\theta^t) \\
		\theta^{t+\varepsilon} &= \theta^t + \varepsilon\rho^{t+\varepsilon/2} \\
		\rho^{t+\varepsilon} &= \rho^{t+\varepsilon/2} + (\varepsilon/2)\nabla_\theta\log p(\theta^{t+\varepsilon}).
	\end{align*}
	That is, within each iteration, two half-steps are taken to update the momentum, while the position is updated once. This procedure generates a candidate draw $\theta'$, which is either accepted or rejected according to a Metropolis acceptance step.
	
	The acceptance probability in the Metropolis step is $\min\left(1, \frac{p(\theta', \rho')}{p(\theta, \rho)}\right)$. By comparing this to the definition of the Hamiltonian in \eqref{hamiltonian}, we can see that the acceptance probability is in fact a function of the change in the Hamiltonian between the proposed point and the starting point, as
	\begin{align*}
		\frac{p(\theta', \rho')}{p(\theta, \rho)} = \exp(H(\theta,\rho) - H(\theta', \rho')).
	\end{align*}
	In an idealised world this change is zero, and thus every proposed draw would be accepted. However, as we are simulating the Hamiltonian dynamics numerically, some error will be introduced. In practice, these errors tend not to be large, and the algorithm achieves high acceptance rates even for very long trajectories.
    
	Altogether, Hamiltonian Monte Carlo can be seen as a variant of the popular Metropolis-Hastings (MH) algorithm with a complicated proposal distribution. The long paths which result from exploiting gradient information allow HMC to mitigate the random walk behavior present in the MH algorithm.
	
	To implement an HMC sampler in practice, in this paper, we use the Stan software package \parencite{stanreference}, which implements an HMC variant known as the No-U-Turn sampler (NUTS), developed by \textcite{hoffman2014}. This algorithm tunes the number of steps of the leapfrog integrator dynamically by doubling the length of the path of the leapfrog integrator until the path starts to make a U-turn, at which point the expansion is stopped. Furthermore Stan implements procedures for adapting both $\varepsilon$ and $M$ during the warmup phase of the chain, thus eliminating the need to hand-tune the values. $M$ can be specified to be either diagonal or a dense matrix, and the theoretically optimal value for $M$ equals the inverse of the covariance matrix of the target distribution \parencite[Section 4.3]{betancourt2018}. For numerical reasons, practical implementations of HMC tend to work with the inverse metric $M^{-1}$ directly.
	
	Implementation of HMC requires the calculation of the gradient of the target log-density. In practically interesting models deriving the analytical expressions might be prohibitively time-consuming, and also prone to errors. Stan solves this problem by using automatic differentiation routines, which allow for exact gradients to be evaluated efficiently, without relying on symbolic methods. In practice, one only needs to be able to write out the posterior density to use Stan.

	\section{Non-centered parametrization} \label{app:non-centered}
	The N-IW distribution,
	\begin{align*}
		\Sigma\mid Y &\sim IW(\nu, S) \\
		\vec(A) \mid \Sigma, Y &\sim N(\vec(\Phi), \Sigma\otimes\Omega),
	\end{align*}
	is essentially a case of a hierarchical model, as the local variances of the elements of $A$ depend on the global parameters in $\Sigma$. This class of models can be difficult for MCMC algorithms to explore, as discussed by \textcite{betancourt2013}. The reason is that changes in the values of $\Sigma$ induce large changes in the curvature of the posterior. For low values, the probability mass concentrates around the mean of $A$, while for large values the density spreads over a larger volume. This is also known as a funnel geometry.
	
	The varying curvature means that a Markov chain is not able to explore the posterior efficiently, as exploration of the narrow part of the posterior requires a very small step size, which is inefficient in the wider part of the distribution. Additionally, the covariances between the elements of $A$ are another source of inefficiency.  In the special case of only impact restrictions, we can follow \textcite{betancourt2013} and use the non-centered parameterization. This means that we declare an auxiliary variable $Z$ that has the same dimension as $A$ and is independent of $\Sigma$. We then sample from the distribution
	\begin{align*}
		\Sigma\mid Y &\sim IW(\nu, S) \\
		\vec(Z) \mid \Sigma, Y &\sim N(0, I).
	\end{align*}
	By the properties of the normal distribution, we obtain the desired distribution on $A$ by setting, for each draw,
	\begin{align*}
		\vec(A) &= \vec(\Phi) + (L_\Sigma \otimes L_{\Omega})\vec(Z) \\
		&= \vec(\Phi) + \vec(L_{\Omega} Z L_\Sigma').
	\end{align*}

	This parametrization alleviates the challenging geometry of the original posterior, as the part corresponding to the VAR matrices is standard normal, and thus the challenging funnel geometry is avoided completely. For an HMC algorithm, it is thus possible to run the leapfrog integrator with a larger stepsize and thus for fewer steps.  Furthermore, the evaluation of the density function of the standard normal distribution is much faster than the evaluation of a general multivariate normal distribution, which saves a lot on computation.

    Note that $B$ and $Z$ are independent because $\Sigma$ and the auxiliary parameter $Z$ are independent and the structural impact matrix $B$ is a function of $\Sigma$ but not $Z$. For HMC this means that the inverse metric $M^{-1}$ (see Appendix \ref{app:algorithm}) is block-diagonal, as there are no covariances between $B$ and $Z$ to be taken into account. In our experience, this makes a diagonal inverse metric a viable choice when a non-centered parameterization is used. Especially in a large model, estimating the inverse metric $M^{-1}$ is usually prohibitively time consuming, and sampling efficiency can be reduced if its estimate is far from the optimal value of the covariance matrix of the target posterior distribution \parencite[Section 4.3]{betancourt2018}. Therefore, in the application of Section \ref{sec:35variable} involving the 35-variable model, we use a diagonal inverse metric. 
    In the oil market application in Section \ref{sec:oilmarket} we use a dense inverse metric because the distribution of $(\Psi_1\dots,\Psi_{12},\tA_{13},\dots, \tA_{24})$ is not a Gaussian, which precludes a non-centered parameterization.

\end{appendices}
	
\end{document}